\documentclass[aps,pra,superscriptaddress,twocolumn,footinbib,showpacs,10pt]{revtex4-1}
\usepackage{fix-cm}
\usepackage[dvipsnames]{xcolor}
\usepackage{xparse}

\usepackage[noamssymbols]{newpxmath} 

\usepackage{amsthm,amsmath,amssymb,amsbsy}

\usepackage{color}
\definecolor{darkred}{rgb}{0.8,0.1,0.1}
\definecolor{lightblue}{rgb}{0.1,0.1,0.8}
\usepackage{ulem}
\usepackage{graphicx,epic,eepic,epsfig,verbatim,color}
 \usepackage{tikz}
\usepackage[T1]{fontenc}
\usepackage[utf8]{inputenc}
\usepackage[breaklinks]{hyperref}
\usepackage{color, soul}
\usepackage{xcolor}
\usepackage[most,breakable]{tcolorbox}

\definecolor{myhlcolor}{rgb}{1, 1, 0}
\definecolor{myhlcolortwo}{rgb}{1, 1, 0}
\definecolor{myhlcolor}{rgb}{1, 1, 1}
\newtheorem{theorem}{Theorem}

\newtheorem{corollary}[theorem]{Corollary}
\newtheorem{definition}[theorem]{Definition}

\newtheorem{proposition}[theorem]{Proposition}

\usepackage{cleveref}
\usepackage{graphicx}
\usepackage{xcolor}

\definecolor{darkblue}{RGB}{0,76,156}
\definecolor{darkkblue}{RGB}{0,0,153}
\definecolor{blue2}{RGB}{102,178,255}

\def\endenv{\ifmmode\;\else{\unskip\nobreak\hfil
\penalty50\hskip1em\null\nobreak\hfil\;
\parfillskip=0pt\finalhyphendemerits=0\endgraf}\fi}

\newenvironment{remark}{\noindent \textbf{{Remark~}}}{}

\mathchardef\ordinarycolon\mathcode`\:
\mathcode`\:=\string"8000
\def\vcentcolon{\mathrel{\mathop\ordinarycolon}}
\begingroup \catcode`\:=\active
  \lowercase{\endgroup
  \let :\vcentcolon
  }

\makeatletter
\def\resetMathstrut@{%
  \setbox\z@\hbox{%
    \mathchardef\@tempa\mathcode`\[\relax
    \def\@tempb##1"##2##3{\the\textfont"##3\char"}%
    \expandafter\@tempb\meaning\@tempa \relax
  }%
  \ht\Mathstrutbox@\ht\z@ \dp\Mathstrutbox@\dp\z@}
\makeatother
\begingroup
  \catcode`(\active \xdef({\left\string(}
  \catcode`)\active \xdef){\right\string)}
\endgroup
\mathcode`(="8000 \mathcode`)="8000

\newcommand{\nc}{\newcommand}
\nc{\rnc}{\renewcommand}
\nc{\beg}{\begin{equation}}
\nc{\eeq}{{\end{equation}}}
\nc{\beqa}{\begin{eqnarray}}
\nc{\eeqa}{\end{eqnarray}}
\nc{\lbar}[1]{\overline{#1}}
\nc{\ketbra}[2]{|#1\rangle\!\langle#2|}

\nc{\avg}[1]{\langle#1\rangle}
\nc{\Rank}{\operatorname{Rank}}
\nc{\smfrac}[2]{\mbox{$\frac{#1}{#2}$}}
\nc{\ox}{\otimes}
\nc{\dg}{\dagger}
\nc{\dn}{\downarrow}
\nc{\cA}{{\cal A}}
\nc{\cB}{{\cal B}}
\nc{\cC}{{\cal C}}
\nc{\cD}{{\cal D}}
\nc{\cE}{{\cal E}}
\nc{\cF}{{\cal F}}
\nc{\cG}{{\cal G}}
\nc{\cH}{{\cal H}}
\nc{\cI}{{\cal I}}
\nc{\cJ}{{\cal J}}
\nc{\cK}{{\cal K}}
\nc{\cL}{{\cal L}}
\nc{\cM}{{\cal M}}
\nc{\cN}{{\cal N}}
\nc{\cO}{{\cal O}}
\nc{\cP}{{\cal P}}
\nc{\cQ}{{\cal Q}}
\nc{\cR}{{\cal R}}
\nc{\cS}{{\cal S}}
\nc{\cT}{{\cal T}}
\nc{\cV}{{\cal V}}
\nc{\cU}{{\cal U}}
\nc{\cX}{{\cal X}}
\nc{\cY}{{\cal Y}}
\nc{\cZ}{{\cal Z}}
\nc{\cW}{{\cal W}}
\nc{\csupp}{{\operatorname{csupp}}}
\nc{\qsupp}{{\operatorname{qsupp}}}
\nc{\var}{{\operatorname{var}}}
\nc{\rar}{\rightarrow}
\nc{\lrar}{\longrightarrow}
\nc{\polylog}{{\operatorname{polylog}}}
\nc{\1}{{\mathds{1}}}
\nc{\wt}{{\operatorname{wt}}}
\nc{\av}[1]{{\left\langle {#1} \right\rangle}}
\nc{\supp}{{\operatorname{supp}}}

\def\x{\xi}

\nc{\RR}{{{\mathbb R}}}
\nc{\CC}{{{\mathbb C}}}
\nc{\FF}{{{\mathbb F}}}
\nc{\NN}{{{\mathbb N}}}
\nc{\ZZ}{{{\mathbb Z}}}
\nc{\PP}{{{\mathbb P}}}
\nc{\QQ}{{{\mathbb Q}}}
\nc{\UU}{{{\mathbb U}}}
\nc{\EE}{{{\mathbb E}}}

\nc{\CHSH}{{\operatorname{CHSH}}}

\nc{\be}{\begin{equation}}
\nc{\ee}{{\end{equation}}}
\nc{\bea}{\begin{eqnarray}}
\nc{\eea}{\end{eqnarray}}
\nc{\Hom}[2]{\mbox{Hom}(\CC^{#1},\CC^{#2})}
\nc{\rU}{\mbox{U}}

\nc{\ob}[1]{#1}

\nc{\SEP}{{\text{SEP}}}
\nc{\NS}{{\text{NS}}}
\nc{\LOCC}{{\text{LOCC}}}
\nc{\PPT}{{\text{PPT}}}
\nc{\EXT}{{\text{EXT}}}
\nc{\Sym}{{\operatorname{Sym}}}

\nc{\ERLO}{{E_{\text{r,LO}}}}
\nc{\ERLOCC}{{E_{\text{r,LOCC}}}}
\nc{\ERPPT}{{E_{\text{r,PPT}}}}
\nc{\ERLOCCinfty}{{E^{\infty}_{\text{r,LOCC}}}}
\nc{\Aram}{{\operatorname{\sf A}}}

\usepackage{tikz}
\usetikzlibrary{arrows}

\makeatletter
\def\grd@save@target#1{%
  \def\grd@target{#1}}
\def\grd@save@start#1{%
  \def\grd@start{#1}}
\tikzset{
  grid with coordinates/.style={
    to path={%
      \pgfextra{%
        \edef\grd@@target{(\tikztotarget)}%
        \tikz@scan@one@point\grd@save@target\grd@@target\relax
        \edef\grd@@start{(\tikztostart)}%
        \tikz@scan@one@point\grd@save@start\grd@@start\relax
        \draw[minor help lines,magenta] (\tikztostart) grid (\tikztotarget);
        \draw[major help lines] (\tikztostart) grid (\tikztotarget);
        \grd@start
        \pgfmathsetmacro{\grd@xa}{\the\pgf@x/1cm}
        \pgfmathsetmacro{\grd@ya}{\the\pgf@y/1cm}
        \grd@target
        \pgfmathsetmacro{\grd@xb}{\the\pgf@x/1cm}
        \pgfmathsetmacro{\grd@yb}{\the\pgf@y/1cm}
        \pgfmathsetmacro{\grd@xc}{\grd@xa + \pgfkeysvalueof{/tikz/grid with coordinates/major step}}
        \pgfmathsetmacro{\grd@yc}{\grd@ya + \pgfkeysvalueof{/tikz/grid with coordinates/major step}}
        \foreach \x in {\grd@xa,\grd@xc,...,\grd@xb}
        \node[anchor=north] at (\x,\grd@ya) {\pgfmathprintnumber{\x}};
        \foreach \y in {\grd@ya,\grd@yc,...,\grd@yb}
        \node[anchor=east] at (\grd@xa,\y) {\pgfmathprintnumber{\y}};
      }
    }
  },
  minor help lines/.style={
    help lines,
    step=\pgfkeysvalueof{/tikz/grid with coordinates/minor step}
  },
  major help lines/.style={
    help lines,
    line width=\pgfkeysvalueof{/tikz/grid with coordinates/major line width},
    step=\pgfkeysvalueof{/tikz/grid with coordinates/major step}
  },
  grid with coordinates/.cd,
  minor step/.initial=.2,
  major step/.initial=1,
  major line width/.initial=2pt,
}
\makeatother

\tikzset{
  treenode/.style = {align=center, inner sep=0pt, text centered,
    font=\sffamily},
  arn_n/.style = {treenode, circle, white, font=\sffamily\bfseries, draw=black,
    fill=black, text width=1.5em},
  arn_r/.style = {treenode, circle, red, draw=red, 
    text width=1.5em, very thick},
  arn_x/.style = {treenode, rectangle, draw=black,
    minimum width=0.5em, minimum height=0.5em}
}

\usepackage{pifont}
\usepackage{graphicx}
\usepackage{bm,bbm}%
\usepackage{braket}
\usepackage{enumerate}
\usepackage{subcaption,caption}
\usepackage{booktabs,multirow}
\usepackage{mathtools,bbding}
\usepackage[percent]{overpic}
\usepackage{microtype}

\usepackage{newpxtext,newpxmath}

\usepackage[most,breakable]{tcolorbox}

\AtBeginDocument{
\heavyrulewidth=.08em
\lightrulewidth=.05em
\cmidrulewidth=.03em
\belowrulesep=.65ex
\belowbottomsep=0pt
\aboverulesep=.4ex
\abovetopsep=0pt
\cmidrulesep=\doublerulesep
\cmidrulekern=.5em
\defaultaddspace=.5em
}

\usepackage{etoolbox}
\hypersetup{linktoc=all}

\DeclareMathOperator{\Tr}{Tr}

\nc{\MIO}{{\text{\rm MIO}}}
\nc{\DIO}{{\text{\rm DIO}}}
\nc{\SIO}{{\text{\rm SIO}}}
\nc{\IO}{{\text{\rm IO}}}

\let\oldproofname\proofname
\renewcommand{\proofname}{\rm\bf{\oldproofname}}

\makeatletter
\renewenvironment{proof}[1][\proofname]{%
  \vspace{-\topsep}
  \pushQED{\qed}
  \normalfont
  \topsep6\p@\@plus6\p@\relax
  \trivlist\item[\hskip\labelsep\bfseries#1\@addpunct{.}]\ignorespaces}{\popQED\endtrivlist\@endpefalse}
\makeatother

\begin{document}
\title{Family of two-parameter multipartite entanglement measures}
\author{Yu Luo}
\email{penroseluoyu@gmail.com}
\affiliation{School of Artificial Intelligence and Computer Science, Shaanxi Normal University, Xi'an, 710062, China}
\author{Zhihua Guo}
\affiliation{School of Mathematics and Statistics, Shaanxi Normal University, Xi’an 710062, China}
\author{Fanxu Meng}
\affiliation{College of Artificial Intelligence, Nanjing Tech University, Nanjing, 211800, China}
\author{Chen-Ming Bai}
\email{baichm@stdu.edu.cn }
\affiliation{Department of Mathematics and Physics, Shijiazhuang Tiedao University,
Shijiazhuang, 050043, China}
%
\begin{abstract}
Multipartite entanglement is regarded as a crucial physical resource in quantum network communication. However, due to the intrinsic complexity of quantum many-body systems, identifying a multipartite entanglement measure that is both efficiently computable and capable of accurately characterizing entanglement remains a challenging problem. 
To address these issues, we propose a family of two-parameter multipartite entanglement measures for mixed states, termed unified-entropy concentratable entanglements. Many well-known multipartite entanglement measures are recovered as special cases of this family of measures, such as the entanglement of formation and the concentratable entanglements introduced by Beckey $et$ $al$. [Phys. Rev. Lett. 127, 140501 (2021)].
We demonstrate that the unified-entropy concentratable entanglements constitute well-defined entanglement monotones and establish several desirable properties they satisfy, such as subadditivity and continuity. 
We further investigate the ordering relations of unified-entropy concentratable entanglements and discuss how these quantities can be efficiently estimated on near-term quantum devices.
As an application, we demonstrate that the unified-entropy concentratable entanglements can effectively distinguish between multi-qubit Greenberger–Horne–Zeilinger states and W states. The ordering relations of these entanglement measures are further validated using four-partite star quantum network states and four-qubit Dicke states. Moreover, we find that the unified-entropy concentratable entanglements exhibit greater sensitivity than the original concentratable entanglements in detecting certain four-partite star quantum network states.
\end{abstract}
\date{\today}
\maketitle
\section{Introduction} 
Entanglement is a crucial resource in quantum information processing, playing a vital role in quantum network communication~\cite{Buscemi2012-PhysRevLett.108.200401,Navascues2020-PhysRevLett.125.240505,Horodecki2009}, quantum teleportation~\cite{Bennett1993-PhysRevLett.70.1895,Boschi1998-PhysRevLett.80.1121}, and quantum dense coding~\cite{Bennett1992-PhysRevLett.69.2881,Mattle1996-PhysRevLett.76.4656}. With advances in quantum control technologies, the utility of entanglement has extended to quantum-enhanced sensing~\cite{chalopin2018quantum,Pooser2020_PhysRevLett.124.230504,defienne2024advances,pezze2021entanglement}, quantum communication~\cite{liao2017satellite,ma2012quantum,Valivarthi2020_PRXQuantum.1.020317,hu2023progress}, and demonstrating quantum advantage on near-term quantum devices~\cite{arute2019quantum,king2025beyond}. However, the practical effectiveness of these applications often depends on the degree of entanglement present in the quantum state. 

Bipartite entanglement has been extensively studied and is relatively well understood. In contrast, multipartite entanglement remains less well understood, largely because its structure becomes increasingly complex as the number of subsystems increases. A well-known example illustrating this complexity is the case of three-qubit entangled states, where two inequivalent classes of genuine multipartite entanglement exist: the Greenberger–Horne–Zeilinger (GHZ) states and the W states. These two classes cannot be transformed into each other under stochastic local operations and classical communication~\cite{Dur2000-PhysRevA.62.062314}. Consequently, developing a unified framework for the detection, quantification, and characterization of multipartite entanglement remains a challenging problem~\cite{Meyer2002,Carvalho2004-PhysRevLett.93.230501,Schwaiger2015-PhysRevLett.115.150502,cianciaruso2016accessible,vrana2023family,hamilton2024probing,foulds2021controlled,beckey2021computable,Bai2025,Liu2025-,Mukherjee2025efficient}.
To address the above challenges, the authors of Ref.~\cite{beckey2021computable} introduced a multipartite entanglement measure known as \textit{concentratable entanglements}. This is a family of entanglement measures that depends on the choice of subsystems in a multipartite system. Many well-known entanglement measures, including concurrence~\cite{Wootters1998_PhysRevLett.80.2245}, can be regarded as special cases of concentratable entanglements. Moreover, this measure can be efficiently estimated on near-term quantum devices~\cite{beckey2021computable,Beckey2023_PhysRevA.107.062425}. \textit{A natural question that arises is whether this unified approach to quantifying and characterizing different types of entanglement can be extended to other entanglement measures that are not special cases of concentratable entanglements.} In other words, is it possible to incorporate important bipartite entanglement measures—such as entanglement of formation and R\'{e}nyi entanglement—into a broader multipartite entanglement framework? Such an extension would provide valuable insights into the relationship between bipartite and multipartite entanglement.

In this work we provide an affirmative answer to the above question. Specifically, we construct a family of two-parameter multipartite entanglement measures for mixed states, denoted by $E^{(s)}_{\alpha,\beta}$, based on the unified entropy $S_{\alpha,\beta}$ (with the precise meaning of the parameters defined in Sec.~\ref{sec:Preliminary}). We refer to these measures as \textit{unified-entropy concentratable entanglements}. We prove that $E^{(s)}_{\alpha,\beta}$ not only non-increasing under local operations and classical communication (LOCC), but also satisfies several desirable properties, including subadditivity and continuity. Moreover, the unified-entropy concentratable entanglements establish a unifying framework that, through suitable choices of the parameters $\alpha$ and $\beta$, recovers several well-known bipartite and multipartite entanglement measures, including R\'{e}nyi entanglement~\cite{horodecki1996necessary,Wang2016-PhysRevA.93.022324,brydges2019probing-science}, Tsallis entanglement~\cite{landsberg1998distributions,Luo2016}, squared concurrence~\cite{Wootters1998_PhysRevLett.80.2245}, entanglement of formation~\cite{Bennett1996-PRA.54.3824}, and even the original concentratable entanglements~\cite{beckey2021computable,Beckey2023_PhysRevA.107.062425}. This unified framework enables the detection, quantification, and characterization of a broad class of entanglement measures within multipartite quantum systems. In particular, we employ the unified-entropy concentratable entanglements to establish a sufficient criterion for certifying genuine multipartite entanglement in three-qudit systems. 

We further investigate the ordering relations of the unified-entropy concentratable entanglements $E^{(s)}_{\alpha,\beta}$. Specifically, when $\beta \geq 1$, $E^{(s)}_{\alpha,\beta}$ decreases monotonically with increasing $\alpha$. When the unified-entropy measures reduce to what we refer to as the R\'{e}nyi and Tsallis concentratable entanglements, we also observe the same inverse dependence on the parameter. In particular, we establish a strict ordering between the von Neumann concentratable entanglements, as defined in this work, and the original concentratable entanglements. These ordering relations provide a useful foundation for efficiently estimating unified-entropy concentratable entanglements on near-term quantum devices. Moreover, by evaluating these quantities for four-partite star quantum network states and four-qubit Dicke states, we not only confirm the predicted hierarchy but also find that the unified-entropy concentratable entanglements exhibit greater sensitivity than the original concentratable entanglements in detecting certain four-partite star quantum network states. This highlights the advantage of parametrized multipartite entanglement measures. Finally, we demonstrate that the unified-entropy concentratable entanglements can effectively distinguish between multiqubit GHZ states and W states.

This paper is organized as follows. 
In Sec.~\ref{sec:Preliminary} we introduced the necessary notation and definitions. 
In Sec.~\ref{sec:mainresults} we provide our main results. 
In Sec.~\ref{sec:lowerbound} we investigate the ordering relations of unified-entropy concentratable entanglements and discuss how these quantities can be efficiently estimated on near-term quantum computers.
In Sec.~\ref{sec:examples}, we investigate the ability of $E^{(s)}_{\alpha,\beta}$ to distinguish between GHZ and W states and study the example of a four-partite star quantum network and four-qubits Dicke states.
We summarize our results in Sec.~\ref{sec:conclusion}.

\section{Preliminary Information} \label{sec:Preliminary}
\subsection{Notation}
We first introduce the necessary notations. We consider a Hilbert space $\mathcal{H}$ of finite dimension $d$. Throughout this paper, all Hilbert spaces considered are finite dimensional. The density operators are denoted by lowercase Greek letters such as \(\rho\), \(\sigma\), \(\psi\), \(\phi\) and so on. 
Conventionally, we use $\psi$ and $\phi$ to represent quantum pure states, i.e., $\psi := \ketbra{\psi}{\psi}$ and $\phi := \ketbra{\phi}{\phi}$, and so on. The support set $\supp\rho$ of a density operator $\rho$ is the vector space spanned by eigenvectors of the operator with non-zero eigenvalues. 
We use \(A_0, A_1\), etc., to represent systems and their corresponding Hilbert spaces. The collection of all bounded operators on system $\mathcal{H}$ will be denoted by $\mathcal{B}(\mathcal{H})$. A multipartite separable operation is a completely positive and trace-preserving map
$\Lambda: \mathcal{B}(\mathcal{H}) \longrightarrow \mathcal{B}(\mathcal{H}')$ that admits a Kraus representation of the fully product form~\cite{watrous2018}
\begin{equation}\label{eq:MSEP}
\Lambda(\cdot) = \sum_{k}
M_{k} 
(\cdot)
M_{k}^{\dagger},    
\end{equation}
where 
\begin{equation}
M_k = 
M_{k1} \otimes M_{k2} \otimes \cdots \otimes M_{kn},    
\end{equation}
and each local operator
$M_{ki} : \mathcal{B}(\mathcal{H}_{A_i}) \to \mathcal{B}(\mathcal{H}'_{A_i})$ acts only on subsystem $A_i$.
The Kraus operators satisfy the trace-preserving condition
\begin{equation}
\sum_{k}
M_{k1}^\dagger M_{k1}
\otimes
M_{k2}^\dagger M_{k2}
\otimes
\cdots
\otimes
M_{kn}^\dagger M_{kn}
= I.    
\end{equation}

\subsection{Unified entropy and unified-entropy concentrarable entanglement}
\begin{definition}
    The unified entropy for a quantum state \(\rho\) is defined as~\cite{hu2006generalized,Rastegin2011some,Rastegin2013unified}
    \begin{equation}
        S_{\alpha,\beta}(\rho)=\frac{1}{(1-\alpha)\beta}[(\Tr\rho^\alpha)^{\beta}-1],
    \end{equation}
    where \(\alpha>0,\alpha\neq1\) and \(\beta>0\).
\end{definition}


The unified entropy reduces to several well-known entropy measures in specific parameter regimes:

It tends to the R\'{e}nyi entropy in the limit \(\beta \to 0\)~\cite{renyi1961measures}:
\begin{equation}
    R_{\alpha}(\rho) := \lim_{\beta \to 0} S_{\alpha,\beta}(\rho) = \frac{1}{1 - \alpha} \log_2 \Tr\rho^{\alpha};
\end{equation}

for \(\beta = 1\) it converges to the Tsallis entropy~\cite{tsallis1988possible}
\begin{equation}
    T_{\alpha}(\rho) := S_{\alpha,1}(\rho) = \frac{1}{1 - \alpha} \left( \Tr\rho^{\alpha} - 1 \right);
\end{equation}

for \(\alpha = 2\) and \(\beta = 1\) it reduces to the linear entropy~\cite{Wootters1998_PhysRevLett.80.2245,Osborne2006-PhysRevLett.96.220503,Uhlmann2000-PhysRevA.62.032307}
\begin{equation}
    S_{\text{lin}}(\rho) := S_{2,1}(\rho) = 1 - \Tr\rho^2;
\end{equation}

and in the limit \(\alpha \to 1\) it converges to the von Neumann entropy~\cite{von2018mathematical}
\begin{equation}
    S(\rho) := \lim_{\alpha \to 1} S_{\alpha,\beta}(\rho) = -\Tr(\rho \log_2 \rho).
\end{equation}

Since \( \alpha \to 1 \) and \( \beta \to 0 \) represent two limiting cases, for the sake of convenience, we treat \( \alpha = 1 \) and \( \beta = 0 \) as these limits throughout the remainder of this paper, unless specified otherwise.

It is straightforward to verify that the unified entropy is nonnegative, \(S_{\alpha,\beta}(\rho)\geq0\), and remains invariant under any unitary Th'r, \(S_{\alpha,\beta}(\mathcal{U}(\rho))=S_{\alpha,\beta}(\rho)\). Additionally, it has been proved that unified entropy is concave for \(  0<\alpha\leq1, \alpha\beta\leq1   \) or \( \alpha\geq1, \alpha\beta\geq1 \)~\cite{hu2006generalized} or \(0<\alpha<1, 0\leq\beta\leq1\)~\cite{Rastegin2011some}: \(\sum_kp_kS_{\alpha,\beta}(\rho_k)\leq S_{\alpha,\beta}(\rho)\) where \( \rho=\sum_kp_k\rho_k \). A direct corollary of this property is the Tsallis entropy for \(\alpha>0\), linear entropy and von Neumann entropy are concave functions. 

We next introduce a family of two-parameter multipartite entanglement measure, referred to as \textit{the unified-entropy concentratable entanglements}.
\begin{definition}[unified-entropy concentratable entanglements]
For any $n$-qudit pure state $\ket{\psi}$ and a subset $s \subseteq [n]$ of qudits, the unified-entropy concentratable entanglements for pure states is defined as
\begin{equation}
E^{(s)}_{\alpha,\beta}(|\psi\rangle) = \frac{1}{2^{|s|}} \sum_{\chi \in \mathcal{P}(s)}S_{\alpha,\beta}(\psi_{\chi}),
\end{equation}
where $\psi_{\chi}$ is the joint reduced state of the subsystems corresponding to the elements in $\chi$, and $\ket{\psi}$ is the associated state. Meanwhile, \( S(\psi_{\chi}) = -\Tr(\psi_{\chi} \log_2 \psi_{\chi}) \) denotes the von Neumann entropy of \( \psi_{\chi} \). When $\chi = \emptyset$, we set $S_{\alpha,\beta}(\psi_{\chi})=0$.

For any $n$-qudit mixed state \(\rho\), the unified-entropy concentratable entanglements is defined as
\begin{equation}
E^{(s)}_{\alpha,\beta}(\rho) =\min\sum_ip_iE^{(s)}_{\alpha,\beta}(\ket{\psi_i}) ,
\end{equation}
where the minimum is taken over all the pure state decomposition ${\{p_i,\psi_i\}}$ of $\rho$.
\end{definition}

From its definition, the unified-entropy concentratable entanglements quantifies the average entanglement between all subsystems of size \(s\) and the remainder of the quantum system. In the special case where \(s\) contains only a single element, the measure reduces to the unified-entropy entanglement introduced in Ref.~\cite{san2011unified}. Specifically, one has  
\[
E_{\alpha,\beta}^{(\{j\})}(\rho) = \min \sum_i p_i S_{\alpha,\beta}(\psi_i^j),
\]  
where \(\rho = \sum_i p_i \psi_i\), and \(\psi_i^j = \Tr_{\{\overline{j}\}}(\psi_i)\) is the reduced state of \(\ket{\psi_i}\) on subsystem \(j\).

Since the unified entropy recovers several well-known entropy measures in specific limits—namely, the R\'{e}nyi entropy when \(\beta = 0\), the Tsallis entropy when \(\beta = 1\), the linear entropy when \((\alpha = 2, \beta = 1)\), and the von Neumann entropy when \(\alpha = 1\)—the corresponding unified-entropy entanglement similarly reduces to established entanglement measures. These include the R\'{e}nyi entanglement~\cite{horodecki1996necessary,Wang2016-PhysRevA.93.022324,brydges2019probing-science}, Tsallis entanglement~\cite{landsberg1998distributions,Luo2016}, squared concurrence~\cite{Wootters1998_PhysRevLett.80.2245}, and the entanglement of formation~\cite{Bennett1996-PRA.54.3824}.

Furthermore, for a general subset \(s\), the unified-entropy concentratable entanglements reduces to the recently introduced multipartite entanglement measure known as \textit{concentratable entanglements} \(C^{(s)}\)~\cite{beckey2021computable} when evaluated at \((\alpha = 2, \beta = 1)\), corresponding to the case of linear entropy. In addition, we define the unified-entropy concentratable entanglements, when reduced to the R\'{e}nyi, Tsallis, and von Neumann entropies, as the \textit{R\'{e}nyi concentratable entanglements, Tsallis concentratable entanglements, and von Neumann concentratable entanglements}, denoted by \(\mathcal{R}_{\alpha}^{(s)}\), \(\mathcal{T}^{(s)}_{\alpha}\), and \(E^{(s)}\), respectively.

\section{Properties of unified-entropy concentratable entanglements}\label{sec:mainresults}
In this section, we show that the unified-entropy concentratable entanglements \( E^{(s)}_{\alpha,\beta} \) is non-negative and monotonic under LOCC operations. Moreover, it is non-increasing on average under LOCC for the parameters range \(\mathcal{A}=\{(\alpha,\beta): 0 < \alpha \leq 1, \alpha\beta \leq 1\}\cup\{(\alpha,\beta): \alpha \geq 1, \alpha\beta \geq 1\}\cup\{ (\alpha,\beta): 0 < \alpha < 1, 0 \leq \beta \leq 1\}\). 
In addition, we present several fundamental properties of \( E^{(s)}_{\alpha,\beta} \), including subadditivity and continuity. 
To address the issue of monotonicity under LOCC, we consider a slightly broader class of operations, namely multipartite separable operations. Our main result is stated as follows.

\begin{theorem}\label{the:vNCE}
Let $\rho$ be an arbitrary $n$-qudit state. For any $\alpha > 0$ and $\beta \geq 0$, subsets $s\subset [n]$, the unified-entropy concentratable entanglements $E^{(s)}_{\alpha,\beta}(\rho)$ is a multipartite entanglement monotone under the multipartite separable operations,  
i.e., $E^{(s)}_{\alpha,\beta}(\rho)$ satisfies following properties

(i)~If $\rho=\sum_ip_i\bigotimes_{j=1}^n\psi_{ij}$ is fully separable, $E^{(s)}_{\alpha,\beta}(\rho)=0$; otherwise $E^{(s)}_{\alpha,\beta}(\rho)>0$.

(ii)~$E^{(s)}_{\alpha,\beta}(\rho)$ is nonincreasing under the multipartite separable operations.

Furthermore, for any \((\alpha,\beta)\in\mathcal{A}\), the unified-entropy concentratable entanglements $E^{(s)}_{\alpha,\beta}(\rho)$ is a multipartite entanglement measure under the multipartite separable operations, 
i.e., $E^{(s)}_{\alpha,\beta}(\rho)$ satisfies following property:

(iii)~$E^{(s)}_{\alpha,\beta}(\rho)$ is non-increasing, on average, under multipartite separable operations. 
\end{theorem}
\begin{proof}
We adopt the method in Ref.~\cite{beckey2021computable} to prove these properties.

(i) It is straightforward to observe that \( E^{(s)}_{\alpha,\beta}(\rho) \) is non-negative, as a direct consequence of the non-negativity of the unified-entropy entropy for \(\alpha,\beta>0\). If \(\ket{\psi} = \bigotimes_{j=1}^n \ket{\psi_j}\), i.e., \(\ket{\psi}\) is a fully separable state, then all the reduced states \(\psi_{\chi}\) are pure. Consequently, we have \( E^{(s)}_{\alpha,\beta}(\ket{\psi}) = 0 \). 

For a mixed state $\rho$, since $\rho$ is fully separable, there necessarily exists a pure-state decomposition $\{p_k, \ket{\psi_k}\}$ such that all the pure states $\ket{\psi_k}$ are fully separable. Consequently, we have
\(
E^{(s)}_{\alpha,\beta}(\rho) = 0.
\)

(ii)
We first show that $E^{(s)}_{\alpha,\beta}(\ket{\psi})$ is non-increasing under the pure state $\ket{\psi}$ to mixed-state ensemble under a multipartite separable operation. 

We can rewrite $E^{(s)}_{\alpha,\beta}(\ket{\psi})$ as 
\begin{eqnarray}\label{eq:eq8}
E^{(s)}_{\alpha,\beta}(\ket{\psi}) 
= \frac{1}{2^{|s|}} \sum_{\chi \in \mathcal{P}(s)} S_{\alpha,\beta}(\psi_{\chi}) 
= \frac{1}{2^{|s|}} \sum_{\chi \in \mathcal{P}(s)} 
E_{\alpha,\beta}(\ket{\psi}_{\chi|\overline{\chi}}), 
\end{eqnarray}
where \( \overline{\chi} \) denotes the complement index subset of \( \chi \) and $E_{\alpha,\beta}(\ket{\psi}_{\chi|\overline{\chi}}):=S_{\alpha,\beta}(\psi_{\chi}) $ is the unified-entropy entanglement for the subsystems $\chi$ and $\overline{\chi}$~\cite{san2011unified}. 
Since $E^{(s)}_{\alpha,\beta}(\ket{\psi})$ contains $2^{|s|}$ terms of local unified entropy \(E_{\alpha,\beta}(\ket{\psi}_{\chi|\overline{\chi}})=S_{\alpha,\beta}(\psi_{\chi})\), it suffices to show that the local unified entropy \(S_{\alpha,\beta}\) is non-increasing under the multipartite separable operations for any \(\alpha>0\) and \(\beta\geq0\), across any possible partition.

Let $\Lambda(\cdot)= \sum_k \Lambda_k(\cdot)=\sum_kM_{k}(\cdot) M_{k}^{\dagger}$ represent a multipartite separable operation 
that admits a Kraus representation \{\(M_k\)\} of Eq.~(\ref{eq:MSEP}),
\begin{equation}\label{eq:Krausreprentation}
M_k= M_{k1} \otimes M_{k2} \otimes \cdots \otimes M_{kn}.    
\end{equation}
Since each Kraus operator \( M_k \) can be expressed as the tensor product of \( n \) mutually commuting local operators \( M_{k j} \) (\( j \in [n] \)), the proof that the local unified entropy is non-increasing under multipartite separable operations proceeds in two steps: first, by demonstrating that the local unified entropy is non-increasing under the action of each local operator \( M_{k j} \) [as shown in Eq.~(\ref{eq:local-monotone})], and second, by extending this monotonicity to the full multipartite separable operation.

Therefore, we first prove that the local unified entropy is non-increasing under the action of each local operator.
Let
\(\Lambda^{(j)}(\cdot) = \sum_k M_{k j} (\cdot) M_{k j}^\dagger\)
be a local operation acting nontrivially only on the \( j \)-th subsystem, where the operators \( {M_{k j}} \) satisfy the completeness relation
\(\sum_k M_{k j}^\dagger M_{k j} = I\).
We will show that the action of \( \Lambda^{(j)} \) on an \( n \)-qudit state \( \rho \) does not increase the unified entropy of any reduced state \( \rho_r \) corresponding to a subsystem \( r \subseteq [n] \).
For the \(k\)-th Kraus component of \(\Lambda^{(j)}\), the resulting post-measurement state is
\[
\rho_{j_k} = \frac{1}{p_{j_k}} M_{k j} \rho M_{k j}^\dagger,
\quad
 p_{j_k} = \Tr[M_{k j} \rho M_{k j}^\dagger].
\]
The corresponding reduced density operators on subsystem \(r\) are
\[
\rho_r = \Tr_{\overline{r}}[\rho],
\qquad
 \rho_r^{j_k} = \Tr_{\overline{r}}[\rho_{j_k}],
\]
where \(\overline{r}\) denotes the complement of subsystem \(r\). 
Following Lemma 1 of Ref.~\cite{beckey2021computable} and Theorem 11 in Ref.~\cite{nielsen2001majorization}, the following majorization relation holds for any subsystems \( r \) and \( j \):
\begin{equation}\label{eq:majorazition}
\boldsymbol{\lambda}(\rho_r) \prec \sum_{j_k} p_{j_k} \boldsymbol{\lambda}(\rho_r^{j_k}),\end{equation}
where \( \boldsymbol{\lambda}(X) \) represents the eigenvalue vector of the Hermitian matrix \( X \), sorted in descending order.
Using the fact proven in Appendix~\ref{app:schur-concave} that the unified entropy is Schur-concave 
(i.e. For any quantum states $\sigma$ and $\tau$, \(\boldsymbol{\lambda}(\sigma)\prec\boldsymbol{\lambda}(\tau) \implies S_{\alpha,\beta}(\sigma)\geq S_{\alpha,\beta}(\tau)\)
holds for all \( \alpha > 0 \) and \( \beta \ge 0 \)), Eq.~(\ref{eq:majorazition}) directly implies that
\begin{equation}\label{eq:local-monotone}
S_{\alpha,\beta}(\rho_r) \ge S_{\alpha,\beta}\left(\sum_{j_k} p_{j_k} \rho_{r}^{j_k}\right)=S_{\alpha,\beta}(\Lambda^{(j)}(\rho_r)),
\end{equation}
for any parameters \( \alpha > 0 \) and \( \beta \ge 0 \).

Next, we address the scenario of monotonicity to the full multipartite separable operation. Without loss of generality, let us consider two local operations \( \Lambda^{(j)} \) and \( \Lambda^{(j')} \).
Applying the results above, we obtain the following chain of majorization relations:
\begin{equation}
\boldsymbol{\lambda}(\rho_r)
\prec \sum_{j_k} p_{j_k} \boldsymbol{\lambda}(\rho_r^{j_k})
\prec \sum_{j_k, j'_{k'}} p_{j_k} p_{j'{k'}} \boldsymbol{\lambda}(\rho_r^{j_k j'_{k'}}).
\end{equation}
Here, the second majorization follows from the transitivity of majorization~\cite{nielsen2001majorization} and the relation
\begin{equation}
\boldsymbol{\lambda}(\rho_r^{j_k}) \prec \sum_{j'{k'}} p_{j'{k'}} \boldsymbol{\lambda}(\rho_r^{j_k j'{k'}}),
\end{equation}
which holds for the successive application of the local operation \( \Lambda^{(j')} \) on the post-measurement state \( \rho_r^{j_k} \).
This establishes that the majorization relation is preserved under the sequential action of multiple local operations. 
Therefore, by using Schur-concavity of unified entropy, we conclude that
\begin{equation}\label{eq:weak-monotone}
    S_{\alpha,\beta}(\rho_r)
    \geq
    S_{\alpha,\beta}(\sum_{\boldsymbol{j}}p_{\boldsymbol{j}}\rho_r^{\boldsymbol{j}})
    =
    S_{\alpha,\beta}(\Lambda^{(j')}\circ\Lambda^{(j)}(\rho_r)),
\end{equation}
where \(p_{\boldsymbol{j}}=p_{j_k}p_{j'_{k'}}\) such that \(\sum_{\boldsymbol{j}}p_{\boldsymbol{j}}=1\)
and
\(\rho_r^{\boldsymbol{j}}\) is the reduced state on the $r$-th subsystem conditioned on the \(\boldsymbol{j}:=  p_{j_k j'_{k'} }  \) outcome. 
Therefore, by mathematical induction, we have shown that the local unified entropy \(S_{\alpha,\beta}\) is non-increasing under the multipartite separable operations for any \(\alpha>0\) and \(\beta\geq0\), across any possible partition.

Therefore, from Eq.~(\ref{eq:eq8}), it follows that for any multipartite separable operation \( \Lambda \), the unified-entropy concentratable entanglement \( E^{(s)}_{\alpha,\beta}(\ket{\psi}) \) satisfies
\begin{equation}\label{eq:pure-mixed}
E^{(s)}_{\alpha,\beta}(\ket{\psi}) \geq E^{(s)}_{\alpha,\beta}(\Lambda(\psi)),
\end{equation}
for any \(\alpha>0\) and \(\beta\geq0\).

Finally, we consider the scenario of transforming a mixed state $\rho$ into a mixed-state ensemble $\{q_k, \rho_k\}$ under a multipartite separable operation $\Lambda=\sum_k\Lambda_k$. Suppose $\rho=\sum_ip_i\psi_i$ is the optimal decomposition of $E^{(s)}_{\alpha,\beta}(\rho)$, i.e., $E^{(s)}_{\alpha,\beta}(\rho)=\sum_ip_iE^{(s)}_{\alpha,\beta}(\ket{\psi_i})$. Then, we have 
\begin{eqnarray}
\Lambda(\rho)&=&\nonumber \sum_k\Lambda_k(\rho)
\\&=&\nonumber
\sum_k\Lambda_k(\sum_ip_i\psi_i)
\\&=&\nonumber
\sum_k\sum_ip_i\Lambda_k(\psi_i)
\\&=&
\sum_k\sum_ip_jp_{ik}\rho_{ik},
\end{eqnarray}
where $\rho_{ik}=\frac{1}{p_{ik}}\Lambda_k(\psi_i)$ with $p_{ik}=\Tr[\Lambda_k(\psi_i)]$. Let $\rho_{i}=\sum_kp_{ik}\rho_{ik}$ and $E_{\alpha,\beta}^{(s)}(\rho_i)=\sum_lp_{il}E_{\alpha,\beta}^{(s)}(\ket{\psi_{il}})$ be the optimal decomposition. 
We can see that $\rho_k=\frac{1}{q_k}\Lambda_k(\rho)$ with $q_k=\Tr[\Lambda_k(\rho)]$. 
Thus, we have 
\begin{eqnarray}
E^{(s)}_{\alpha,\beta}(\rho) 
&=&\nonumber 
\sum_ip_iE^{(s)}_{\alpha,\beta}(\ket{\psi_i})
\\&\overset{(a)}{\geq}&\nonumber
\sum_ip_iE^{(s)}_{\alpha,\beta}(\sum_kp_{ik}\rho_{ik})
\\&=&\nonumber
\sum_ip_iE^{(s)}_{\alpha,\beta}(\rho_i)
\\&=&\nonumber
\sum_i\sum_lp_ip_{il}E^{(s)}_{\alpha,\beta}(\ket{\psi_{il}})
\\&\overset{(b)}{\geq}&\nonumber
E^{(s)}_{\alpha,\beta}(\sum_i\sum_lp_ip_{il}\psi_{il})
\\&=&\nonumber
E^{(s)}_{\alpha,\beta}(\sum_kq_k\rho_k)
\\&=&
E^{(s)}_{\alpha,\beta}(\Lambda(\rho)),
\end{eqnarray}
where (a) follow from Eq.~(\ref{eq:pure-mixed}) and $\Lambda(\psi_i)=\sum_kp_{ik}\rho_{ik}$, (b) follows from the convexity of convex-roof extended measure~\cite{vidal2000entanglement}. 


(iii) The proof follows the same reasoning as property~(ii) and consists of two steps. First, we show that during the separable transformation from a pure state to a mixed-state ensemble, the local unified entropy is non-increasing, on average, under each local operator 
\( M_{kj} \) for all parameters \((\alpha,\beta)\in\mathcal{A}\). Second, we establish that for transformations between mixed-state ensembles, the unified-entropy concentratable entanglement is non-increasing, on average, under multipartite separable operations for all parameters 
\((\alpha,\beta)\in\mathcal{A}\). 
The key point of the two steps of the proof lies in utilizing the concavity of the unified entropy within the parameter domain \(\mathcal{A}\)~\cite{hu2006generalized,Rastegin2011some}. By combining this property with Eqs.~(\ref{eq:local-monotone}) and (\ref{eq:weak-monotone}), the inequalities can be further relaxed to
\begin{equation}\label{eq:local-monotone-iii}
S_{\alpha,\beta}(\rho_r)
\ge S_{\alpha,\beta}\left(\sum_{j_k} p_{j_k} \rho_{r}^{j_k}\right)
\ge \sum_{j_k} p_{j_k} S_{\alpha,\beta}\left(\rho_{r}^{j_k}\right),
\end{equation}
and
\begin{equation}\label{eq:weak-monotone-iii}
S_{\alpha,\beta}(\rho_r)
\ge
S_{\alpha,\beta}\left( \sum_{\boldsymbol{j}} p_{\boldsymbol{j}} \rho_r^{\boldsymbol{j}} \right)
\ge
\sum_{\boldsymbol{j}} p_{\boldsymbol{j}} S_{\alpha,\beta}\left( \rho_r^{\boldsymbol{j}} \right).
\end{equation}
Therefore, by combining Eqs.~(\ref{eq:eq8}) and (\ref{eq:weak-monotone-iii}), we establish the counterpart of Eq.~(\ref{eq:pure-mixed}). Specifically, for any multipartite separable operation \( \Lambda \), the unified-entropy concentratable entanglement \( E^{(s)}_{\alpha,\beta}(\ket{\psi}) \) satisfies
\begin{equation}\label{eq:pure-mixed-iii}
E^{(s)}_{\alpha,\beta}(\ket{\psi}) \geq \sum_{k} p_k E^{(s)}_{\alpha,\beta}(\rho_k),
\end{equation}
for any \((\alpha, \beta) \in \mathcal{A}\), where \( p_k = \Tr[\Lambda_k(\psi)] \) and \( \rho_k = \frac{1}{p_k}\Lambda_k(\psi) \).

Finally, we consider the scenario of transforming a mixed state $\rho$ into a mixed-state ensemble $\{q_k, \rho_k\}$ under a multipartite separable operation $\Lambda=\sum_k\Lambda_k$. Let $\rho=\sum_ip_i\psi_i$ be the optimal decomposition of $E^{(s)}_{\alpha,\beta}(\rho)$, i.e., $E^{(s)}_{\alpha,\beta}(\rho)=\sum_ip_iE^{(s)}_{\alpha,\beta}(\ket{\psi_i})$. Similarly, we have 
\begin{eqnarray}
\Lambda(\rho)
= \sum_k\Lambda_k(\sum_ip_i\psi_i)
=\sum_k\sum_ip_ip_{ik}\rho_{ik},
\end{eqnarray}
where $\rho_{ik}=\frac{1}{p_{ik}}\Lambda_k(\psi_i)$ with $p_{ik}=\Tr[\Lambda_k(\psi_i)]$. Let $\rho_{ik}=\sum_lp_{ikl}\psi_{ikl}$ be the optimal decomposition of $E^{(s)}_{\alpha,\beta}(\rho_{ik})$. Since $\rho_k=\frac{1}{q_k}\Lambda_k(\rho)$ with $q_k=\Tr[\Lambda_k(\rho)]$, we then obtain
\begin{eqnarray}
E^{(s)}_{\alpha,\beta}(\rho) &=&\nonumber \sum_ip_iE^{(s)}_{\alpha,\beta}(\ket{\psi_i})
\\&\overset{(a)}{\geq}&\nonumber
\sum_ip_i\sum_kp_{ik}E^{(s)}_{\alpha,\beta}(\rho_{ik})
\\&=&\nonumber
\sum_i\sum_k\sum_lp_ip_{ik}p_{ikl}E^{(s)}_{\alpha,\beta}(\psi_{ikl})
\\&\overset{(b)}{\geq}&
\sum_kq_kE^{(s)}_{\alpha,\beta}(\rho_k),
\end{eqnarray}
where (a) follow from Eq.~(\ref{eq:pure-mixed-iii}), (b) follows from $q_k\rho_k=\sum_{il}p_ip_{ik}p_{ikl}\psi_{ikl}$ with $q_k=\Tr[\sum_{il}p_ip_{ik}p_{ikl}\psi_{ikl}]$ and $\sum_{il}\frac{p_ip_{ik}p_{jkl}}{q_k}E^{(s)}_{\alpha,\beta}(\ket{\psi_{ikl}})\geq E^{(s)}_{\alpha,\beta}(\rho_k)$. 
\end{proof}
Since the set of separable operations contains all LOCC operations, Theorem~\ref{the:vNCE} further implies that for any \(\alpha > 0\) and \(\beta \geq 0\), the unified-entropy concentratable entanglement \(E^{(s)}_{\alpha,\beta}(\rho)\) is non-increasing under LOCC. Hence, it serves as a valid multipartite entanglement monotone. Moreover, when \((\alpha, \beta) \in \mathcal{A}\), \(E^{(s)}_{\alpha,\beta}(\rho)\) is non-increasing, on average, under LOCC operations.

It is worth noting that the parameter domain \(\mathcal{A}\) covers a broad range of entanglement families, including the R\'{e}nyi concentratable entanglements (\(\beta = 0\)) for \(0 < \alpha < 1\), the Tsallis concentratable entanglements (\(\beta = 1\)) for \(\alpha > 0\), the von Neumann concentratable entanglements (\(\alpha = 1\)), and the orginal concentratable entanglements (\(\alpha = 2, \beta = 1\)).

We now show that the unified-entropy concentratable entanglements has the following properties:
For any subsets $s,s'\subset [n]$, the unified-entropy concentratable entanglements satisfies the subadditivity relation:
\begin{proposition}\label{prop:subadd}
For any multipartite pure state $\ket{\psi}$ and the parameters range \( \mathcal{B}:= \{(\alpha,\beta):\alpha\geq1,\beta=1\}\cup\{(\alpha,\beta):\alpha=0,\beta=0\}\), the subadditivity of \(E^{(s)}_{\alpha,\beta}\) holds:
\begin{equation}\label{eq:subadditivity}
E^{(s\cup s')} _{\alpha,\beta}  (\ket{\psi})\leq E^{(s)} _{\alpha,\beta}(\ket{\psi}) 
+ E^{(s')} _{\alpha,\beta}(\ket{\psi}),    
\end{equation}
where $s\cap s'=\emptyset$.     
\end{proposition}

\begin{proof}
It is easy to see that proving $E^{(s\cup s')} _{\alpha,\beta}(\ket{\psi})\leq E^{(s)} _{\alpha,\beta}(\ket{\psi}) + E^{(s')} _{\alpha,\beta}(\ket{\psi})$ is equivalent to proving 

\begin{eqnarray}\label{eq:subaddtive}
2^{|s'|}\cdot\sum_{\chi_s\in\mathcal{P}(s)}S _{\alpha,\beta}(\psi_{\chi_s})
&+&\nonumber  2^{|s|} \cdot\sum_{\chi_{s'}\in\mathcal{P}(s')}S _{\alpha,\beta}(\psi_{\chi_{s'}})
\\&\geq& 
\sum_{\chi\in\mathcal{P}(s)\times\mathcal{P}(s')}S _{\alpha,\beta}(\psi_{\chi}),   
\end{eqnarray}
where we have used that fact that $s\cap s'=\emptyset\implies
\mathcal{P}(s\cup s')\cong \mathcal{P}(s)\times\mathcal{P}(s')$ and $|s \cup s'|=|s|+|s'|$.

To demonstrate the validity of Eq.~(\ref{eq:subaddtive}), we adopt the following proof strategy: First, observe that in this inequality, \(2^{|s'|} \cdot \sum_{\chi_s \in \mathcal{P}(s)} S  _{\alpha,\beta}  (\psi_{\chi_s})\), \(2^{|s|} \cdot \sum_{\chi_{s'} \in \mathcal{P}(s')} S _{\alpha,\beta}
 (\psi_{\chi_{s'}})\), and \(\sum_{\chi \in \mathcal{P}(s) \times \mathcal{P}(s')} S _{\alpha,\beta}(\psi_{\chi})\) each contains \(2^{|s|+|s'|}\) terms of the form \(S _{\alpha,\beta}(\psi_{\chi})\), where \(\chi \in \mathcal{P}(s \cup s')\). Next, we cleverly organize the terms in \(2^{|s'|} \cdot \sum_{\chi_s \in \mathcal{P}(s)} S  _{\alpha,\beta}  (\psi_{\chi_s})\) into a table with \(2^{|s'|}\) rows and \(2^{|s|}\) columns. Similarly, the terms in \(2^{|s|} \cdot \sum_{\chi_{s'} \in \mathcal{P}(s')} S _{\alpha,\beta}(\psi_{\chi_{s'}})\) and \(\sum_{\chi \in \mathcal{P}(s) \times \mathcal{P}(s')} S _{\alpha,\beta}(\psi_{\chi})\) are arranged into such a table. 

In these tables, each cell corresponds to three terms in the form \(S _{\alpha,\beta}(\psi_A)\), \(S _{\alpha,\beta}(\psi_B)\), and \(S_{\alpha,\beta}(\psi_{A \cup B})\), where \(A, B \in \mathcal{P}(s \cup s')\). Finally, by invoking the subadditivity of unified entropy for the range \(\mathcal{B}\), each cell satisfies the non-negative relation~\cite{Nielsen2010,Rastegin2011some,van2002renyi}
\[
S _{\alpha,\beta}(\psi_A) + S _{\alpha,\beta}(\psi_B) \geq S_{\alpha,\beta}(\psi_{A \cup B}),\] completing the proof.

To better illustrate the effectiveness of the proof strategy, we take \( s = \{1, 2, 3\} \) and \( s' = \{4, 5\} \) as an example. Other more general cases can be similarly and directly inferred. As shown in Table~\ref{tab:example}, each cell of the table corresponds to a term \( S_{\chi} := S_{\alpha,\beta}(\psi_{\chi}) \), where \( \chi \in \mathcal{P}(s \cup s') \). The elements in each cell of the first row of the table correspond to each von Neumann entropy \( S_{\chi_s} \) associated with \( \mathcal{P}(s) \), arranged in ascending order based on the cardinality of the subsets. Similarly, the elements in each cell of the first column of the table correspond to each von Neumann entropy \( S_{\chi_{s'}} \) associated with \( \mathcal{P}(s') \). The elements in the remaining cells of Table~\ref{tab:example} are determined by their corresponding rows and columns. For instance, the position of the cell containing \( S_{134} \) is determined by the row corresponding to \( S_4 \) (the second row) and the column corresponding to \( S_{13} \) (the sixth column).

\begin{table}[h!] 
\centering 
\begin{tabular}{|c|c|c|c|c|c|c|c|}
\hline
$S_{\emptyset}$ & $S_{1}$ & $S_{2}$ & $S_{3}$ & $S_{12}$    & $\textcolor{black}{S_{13}}$ & $S_{23}$   & $S_{123}$ \\ \hline
$\textcolor{black}{S_{4}}$         & $S_{14}$  & $S_{24}$   & $S_{34}$    & $S_{124}$    & $\textcolor{red}{S_{134}}$    & $S_{234}$      & $S_{1234}$   \\ \hline
$S_{5}$         & $S_{15}$  & $S_{25}$   & $S_{35}$    & $S_{125}$       & $S_{135}$   &$S_{235}$     & $S_{1235}$  \\ \hline
$S_{45}$        & $S_{145}$ & $S_{245}$  & $S_{345}$ & $S_{1245}$     & $S_{1345}$  & $S_{2345}$    & $S_{12345}$  \\ \hline
\end{tabular}
\caption{Each cell in the table contains one term from the sum \( \sum_{\chi \in \mathcal{P}(s) \times \mathcal{P}(s')} S_{\alpha,\beta}(\psi_{\chi}) \).}
\label{tab:example} 
\end{table}

As shown in Table~\ref{tab:example-2}, for \(2^{|s'|} \cdot \sum_{\chi_s \in \mathcal{P}(s)}  
S _{\alpha,\beta}(\psi_{\chi_s})\), we place all terms \(S_{\alpha,\beta}(\psi_{\chi_s})\) (\(\chi_s \in \mathcal{P}(s)\)), totaling \(2^{|s|}\), in the first row of a table. Each column of the table contains identical entries corresponding to the same term. 

As shown in Table~\ref{tab:example-3}, for \(2^{|s|}\cdot\sum_{\chi_{s'}\in\mathcal{P}(s')}S 
_{\alpha,\beta}(\psi_{\chi_{s'}}) \), we place all terms \(S_{\alpha,\beta}(\psi_{\chi_{s'}})\) (\(\chi_{s'} \in \mathcal{P}(s')\)), totaling \(2^{|s'|}\), in the first column of a table. Each row of the table contains identical entries corresponding to the same term. 

Finally, by comparing the three terms corresponding to each cell in the three tables, the proof of Eq.~(\ref{eq:subaddtive}) can be easily completed using the subadditivity property of unified entropy.

\begin{table}[h!] 
\centering 
\begin{tabular}{|c|c|c|c|c|c|c|c|}
\hline
$S_{\emptyset}$ & $S_{1}$ & $S_{2}$ & $S_{3}$ & $S_{12}$    & $S_{13}$ & $S_{23}$   & $S_{123}$ \\ \hline
$S_{\emptyset}$ & $S_{1}$ & $S_{2}$ & $S_{3}$ & $S_{12}$    & $\textcolor{red}{S_{13}}$ & $S_{23}$   & $S_{123}$  \\ \hline
$S_{\emptyset}$ & $S_{1}$ & $S_{2}$ & $S_{3}$ & $S_{12}$    & $S_{13}$ & $S_{23}$   & $S_{123}$ \\ \hline
$S_{\emptyset}$ & $S_{1}$ & $S_{2}$ & $S_{3}$ & $S_{12}$    & $S_{13}$ & $S_{23}$   & $S_{123}$  \\ \hline
\end{tabular}
\caption{Each cell in the table contains one term from the sum \(2^{|s'|}\cdot\sum_{\chi_s\in\mathcal{P}(s)}S_{\alpha,\beta}(\psi_{\chi_s}) \).}
\label{tab:example-2} 
\end{table}


\begin{table}[h!] 
\centering 
\begin{tabular}{|c|c|c|c|c|c|c|c|}
\hline
$S_{\emptyset}$ & $S_{\emptyset}$ & $S_{\emptyset}$ & $S_{\emptyset}$ & $S_{\emptyset}$    & $S_{\emptyset}$ & $S_{\emptyset}$   & $S_{\emptyset}$ \\ \hline
$S_{4}$         & $S_{4}$ & $S_{4}$ & $S_{4}$ & $S_{4}$    & $\textcolor{red}{S_{4}}$ & $S_{4}$   & $S_{4}$  \\ \hline
$S_{5}$         & $S_{5}$ & $S_{5}$ & $S_{5}$ & $S_{5}$    & $S_{5}$ & $S_{5}$   & $S_{5}$ \\ \hline
$S_{45}$        & $S_{45}$ & $S_{45}$ & $S_{45}$ & $S_{45}$    & $S_{45}$ & $S_{45}$   & $S_{45}$  \\ \hline
\end{tabular}
\caption{Each cell in the table contains one term from the sum \(2^{|s|}\cdot\sum_{\chi_{s'}\in\mathcal{P}(s')}S_{\alpha,\beta}(\psi_{\chi_{s'}}) \).}
\label{tab:example-3} 
\end{table}    
\end{proof}

From Proposition~\ref{prop:subadd}, it follows that within the parameter domain $\mathcal{B}$, the unified-entropy concentratable entanglement of two independent subsystems considered jointly cannot exceed the sum of their individual unified-entropy concentratable entanglements.
It is worth noting that the parameter domain $\mathcal{B}$ encompasses the entire family of Tsallis concentratable entanglements with \(\alpha > 1\), thereby providing a positive answer to Conjecture~1.2 proposed in Ref.~\cite{Liu2025-}.

\begin{proposition}\label{prop:subadditive}
Let $A$ and $B$ be disjoint subsystems, and let $\ket{\psi} = \ket{\psi}_A \otimes \ket{\psi}_B$. For $\alpha > 0$ and $\beta \geq 0$, the unified-entropy concentratable entanglements satisfies the following identity:

\begin{eqnarray} \label{eq:qsi-additive}
E^{(s)}_{\alpha,\beta}(\ket{\psi})
    &=&\nonumber E^{(s\cap A)}_{\alpha,\beta}(\ket{\psi}_A)+E^{(s\cap B)}_{\alpha,\beta}(\ket{\psi}_B)
    \\&+& (1-\alpha)\beta E^{(s\cap A)}_{\alpha,\beta}(\ket{\psi}_A)E^{(s\cap B)}_{\alpha,\beta}(\ket{\psi}_B).
\end{eqnarray}   

In particular, when $\alpha = 1$ or $\beta = 0$, the unified-entropy concentratable entanglements reduces to the von Neumann concentratable entanglements or the R\'{e}nyi concentratable entanglements, respectively. In these cases, the equality simplifies to:
\begin{equation}\label{eq:vn-re-add-1}
E^{(s)}_{\alpha,\beta}(\ket{\psi})
= E^{(s \cap A)}_{\alpha,\beta}(\ket{\psi}_A) + E^{(s \cap B)}_{\alpha,\beta}(\ket{\psi}_B).
\end{equation}

Furthermore, when $s = [n]$, the expression becomes:
\begin{eqnarray}
E^{(s)}_{\alpha,\beta}(\ket{\psi})
    &=&\nonumber E^{(A)}_{\alpha,\beta}(\ket{\psi}_A)+E^{(B)}_{\alpha,\beta}(\ket{\psi}_B)
    \\&+& (1-\alpha)\beta E^{(A)}_{\alpha,\beta}(\ket{\psi}_A)E^{(B)}_{\alpha,\beta}(\ket{\psi}_B).
\end{eqnarray}
and, in the case $\alpha = 1$ or $\beta = 0$, further reduces to:
\begin{equation}\label{eq:vn-ren-add-2}
E^{(s)}_{\alpha,\beta}(\ket{\psi})
= E^{(A)}_{\alpha,\beta}(\ket{\psi}_A) + E^{(B)}_{\alpha,\beta}(\ket{\psi}_B).
\end{equation}
\end{proposition}
\begin{proof}
We first discuss the case of \(\alpha=1\) or \(\beta=0\). In this case, we have
\begin{eqnarray}
E^{(s)}_{\alpha,\beta}(\ket{\psi}) &=&\nonumber \frac{1}{2^{|s|}} \sum_{\chi \in \mathcal{P}(s)} S_{\alpha,\beta}(\psi_{\chi})  
\\\nonumber  &=&
\frac{1}{2^{|s|}} \sum_{\chi \in \mathcal{P}(s)} S_{\alpha,\beta}(\psi_{\chi\cap A}\otimes\psi_{\chi\cap B})  
\\\nonumber  &=&
\frac{1}{2^{|s|}} \sum_{\chi \in \mathcal{P}(s)} \left(S_{\alpha,\beta}(\psi_{\chi\cap A})+  S_{\alpha,\beta}(\psi_{\chi\cap B})\right)
\\\nonumber  &=&
\frac{1}{2^{|s|}} \sum_{\chi_a \in \mathcal{P}(s\cap A)} S_{\alpha,\beta}(\psi_{\chi_a})+\frac{1}{2^{|s|}} \sum_{\chi_b \in \mathcal{P}(s\cap B)} S_{\alpha,\beta}(\psi_{\chi_b})
\\  &=&
E^{(s\cap A)}_{\alpha,\beta}(\ket{\psi}_A)+E^{(s\cap B)}_{\alpha,\beta}(\ket{\psi}_B),
\end{eqnarray} 
where the third line holds follow from the additivity of von Neumann entropy (\(\alpha=1\)) and R\'{e}nyi entropy (\(\beta=0\)), the fourth line holds follows from $A\cup B=s, A\cap B=\emptyset\implies \mathcal{P}(s)\cong \mathcal{P}(s\cap A)\times\mathcal{P}(s\cap B)$ and $|s|=|s\cap A|+|s\cap B|$.    

For $\alpha > 0$, $\alpha \neq 1$, and $\beta > 0$, we first note that the following relation holds:
\begin{widetext}
\begin{eqnarray}
E^{(s)}_{\alpha,\beta}(\ket{\psi}) &=&\nonumber \frac{1}{2^{|s|}} \sum_{\chi \in \mathcal{P}(s)} S_{\alpha,\beta}(\psi_{\chi})  
\\\nonumber  &=&
\frac{1}{2^{|s|}} \sum_{\chi \in \mathcal{P}(s)} S_{\alpha,\beta}(\psi_{\chi\cap A}\otimes\psi_{\chi\cap B})  
\\\nonumber  &=&
\frac{1}{2^{|s|}} \sum_{\chi \in \mathcal{P}(s)} \frac{1}{(1-\alpha)\beta}
\left( \Tr(\psi_{\chi\cap A}^{\alpha}\otimes\psi_{\chi\cap B}^{\alpha})^{\beta}-1\right)
\\\nonumber  &=&
\frac{1}{2^{|s|}} \sum_{\chi \in \mathcal{P}(s)} \frac{1}{(1-\alpha)\beta}\left( (\Tr\psi_{\chi\cap A}^{\alpha})^{\beta}(\Tr\psi_{\chi\cap B}^{\alpha})^{\beta}-1\right)
\\\nonumber  &=&
\frac{1}{2^{|s|}}\frac{1}{(1-\alpha)\beta}
\sum_{\chi_a \in \mathcal{P}(s\cap A)} (\Tr\psi_{\chi_a}^{\alpha})^{\beta}
\cdot
\sum_{\chi_b \in \mathcal{P}(s\cap B)}(\Tr\psi_{\chi_b}^{\alpha})^{\beta}
-\frac{1}{(1-\alpha)\beta}.
\end{eqnarray}     
\end{widetext}
Moreover, noting that
\[
\sum_{\chi_X \in \mathcal{P}(s \cap X)} (\Tr\psi_{\chi_X}^{\alpha})^{\beta} = 2^{|s \cap X|} (1 - \alpha) \beta \, E^{(s \cap X)}_{\alpha,\beta}(\ket{\psi})
\]
holds for  \(X = A, B\),
and substituting this expression into the equation above yields the simplified form:
\begin{eqnarray}
E^{(s)}_{\alpha,\beta}(\ket{\psi})
    &=&\nonumber E^{(s\cap A)}_{\alpha,\beta}(\ket{\psi}_A)+E^{(s\cap B)}_{\alpha,\beta}(\ket{\psi}_B)
    \\&+& (1-\alpha)\beta E^{(s\cap A)}_{\alpha,\beta}(\ket{\psi}_A)E^{(s\cap B)}_{\alpha,\beta}(\ket{\psi}_B).
\end{eqnarray}
\end{proof}

\begin{remark}
A special case of Proposition \ref{prop:subadditive} is obtained when \(\alpha=2\) and \(\beta=1\). In this case, 
\begin{eqnarray}
E^{(s)}_{\alpha,\beta}(\ket{\psi})
&=&\nonumber E^{(s \cap A)}_{\alpha,\beta}(\ket{\psi}_A)
+ E^{(s \cap B)}_{\alpha,\beta}(\ket{\psi}_B)
\\&-& E^{(s \cap A)}_{\alpha,\beta}(\ket{\psi}) E^{(s \cap B)}_{\alpha,\beta}(\ket{\psi}),   
\end{eqnarray}
which recovers a result of concentratable entanglements $C^{(s)}$ previously reported in Ref.~\cite{Schatzki2024-PhysRevResearch.6.023019}.    
\end{remark}

As two immediate corollaries of Proposition~\ref{prop:subadditive}, we obtain the following:

\begin{corollary}
Let $0 < \alpha < 1$ and $\beta > 0$. Then the unified-entropy concentratable entanglements satisfies the following superadditivity property:
\begin{eqnarray}
E^{(s)}_{\alpha,\beta}(\ket{\psi})
    &\geq&\nonumber E^{(s\cap A)}_{\alpha,\beta}(\ket{\psi}_A)+E^{(s\cap B)}_{\alpha,\beta}(\ket{\psi}_B).
\end{eqnarray}
Conversely, for $\alpha > 1$ and $\beta > 0$, the following subadditivity property holds:
\begin{eqnarray}
E^{(s)}_{\alpha,\beta}(\ket{\psi})
    &\leq&\nonumber E^{(s\cap A)}_{\alpha,\beta}(\ket{\psi}_A)+E^{(s\cap B)}_{\alpha,\beta}(\ket{\psi}_B).
\end{eqnarray}
\end{corollary}

\begin{corollary}
Let $s = [n]$. For $\alpha = 1$ or $\beta = 0$, the unified-entropy concentratable entanglements is additive under tensor powers, i.e.,
\begin{equation}
E^{(s)}_{\alpha,\beta}(\ket{\psi}^{\otimes k}) = k E^{(s)}_{\alpha,\beta}(\ket{\psi}),
\end{equation}    
for any $k \in \mathbb{N}$.
\end{corollary}

The above results provide a complete characterization of the conditions under which the unified-entropy concentratable entanglements exhibits superadditivity, subadditivity, or additivity. In particular, the fact that both the von Neumann concentratable entanglements and the R\'{e}nyi concentratable entanglements satisfy additivity (as shown in Eqs.(\ref{eq:vn-re-add-1})-(\ref{eq:vn-ren-add-2})) indicates that these two multipartite entanglement measures possess more desirable properties than the concentratable entanglements $C^{(s)}$~\cite{Schatzki2024-PhysRevResearch.6.023019}.
\begin{proposition}\label{the:RCE-property}
Let $s = [n]$. A three-qudit state $\ket{\psi}_{ABC}$ is genuinely multipartite entangled if
\begin{equation}
E^{(s)}_{\alpha,\beta}(\ket{\psi}_{ABC}) >
\begin{cases}
\displaystyle \frac{1}{2} \log_2 d, & \text{if } \alpha = 1 \text{ or } \beta = 0, \\
\displaystyle \frac{1}{2(1 - \alpha)\beta} \left[d^{(1 - \alpha)\beta} - 1\right], & \text{if } \alpha \neq 1 \text{ and } \beta \neq 0.
\end{cases} 
\end{equation}
\end{proposition}
\begin{proof}
It is easy to see that $E^{(s)}=0$ for a product state. Therefore, without loss of generality, we assume \( \ket{\psi}_{ABC} = \ket{\phi}_{AB} \ket{\phi}_C \). The following chain of equalities holds:
\begin{eqnarray}\label{eq:genuninelymultipartite}
    E^{(s)}_{\alpha,\beta}(\ket{\psi}_{ABC}) &=& \nonumber \frac{1}{2^3}\sum_{\chi\in\mathcal{P}(s)}S_{\alpha,\beta}(\psi_{\chi})
\\ &\overset{(a)}{=}& \nonumber
    \frac{1}{4}(S_{\alpha,\beta}(\psi_A)+S_{\alpha,\beta}(\psi_B)+S_{\alpha,\beta}(\psi_C))
\\ &\overset{(b)}{=}& \nonumber
    \frac{1}{4}(S_{\alpha,\beta}(\psi_A)+S_{\alpha,\beta}(\psi_B))      
\\ &\overset{(c)}{=}& \nonumber
    \frac{1}{4}(S_{\alpha,\beta}(\phi_A)+S_{\alpha,\beta}(\phi_B))   
\\ &=& 
    E^{(s)}_{\alpha,\beta}(\ket{\phi}_{AB}),  
\end{eqnarray}
where 
(a) holds follow from $S_{\alpha,\beta}(\psi_{\chi})=S _{\alpha,\beta}(\psi_{\overline{\chi}})$ for any $\chi\in \mathcal{P}(s)$, 
(b) holds follow from $S_{\alpha,\beta}(\psi_C)=0$, 
(c) holds by noting that $\ket{\psi}_{ABC} = \ket{\phi}_{AB} \ket{\phi}_C$. 
Since $E^{(s)}_{\alpha,\beta}(\ket{\phi}_{AB})=\frac{1}{2}S_{\alpha,\beta}(\phi_A)$, the maximal value of  $E^{(s)}_{\alpha,\beta}(\ket{\phi}_{AB})$ is~\cite{hu2006generalized} 
\begin{equation}
E^{(s)}_{\alpha,\beta}(\ket{\phi}_{AB}) =
\begin{cases}
\displaystyle \frac{1}{2} \log_2 d, & \text{if } \alpha =1 \text{ or } \beta =0, \\
\displaystyle \frac{1}{2} \cdot \frac{1}{(1 - \alpha)\beta} \left[d^{(1 - \alpha)\beta} - 1\right], & \text{if } \alpha \neq 1 \text{ and } \beta \neq 0.
\end{cases}
\end{equation}
Thus, we conclude that \( \ket{\psi}_{ABC} \) is genuinely multipartite entangled if 
\begin{equation}
E^{(s)}_{\alpha,\beta}(\ket{\psi}_{ABC}) >
\begin{cases}
\displaystyle \frac{1}{2} \log_2 d, & \text{if } \alpha =1 \text{ or } \beta =0, \\
\displaystyle \frac{1}{2} \cdot \frac{1}{(1 - \alpha)\beta} \left[d^{(1 - \alpha)\beta} - 1\right], & \text{if } \alpha \neq 1 \text{ and } \beta \neq 0.
\end{cases}
\end{equation}
\end{proof}
A special case of Proposition~\ref{the:RCE-property} is obtained when $\alpha = 2$ and $\beta = 1$. In this case, $E^{(s)}_{2,1}(\ket{\phi}_{AB}) = C^{(s)}(\ket{\phi}_{AB}) > \frac{1}{4}$. This result recovers a result presented in Ref.~\cite{schatzki2021entangled}.

\begin{theorem}\label{the:contiunity}
Let \( \ket{\psi} \) and \( \ket{\phi} \) be two pure states such that their trace distance satisfies \( D(\psi, \phi) := \frac{1}{2}\|\psi - \phi\|_1 \leq \epsilon < \frac{1}{2} \), where \( D(\cdot, \cdot) \) denotes the trace distance. Then, for parameters \( \alpha > 1 \) and \( \beta \geq 1 \), we have the following continuity bound:
\begin{equation}\label{eq:contin_1}
\left|E^{(s)}_{\alpha,\beta}(\ket{\psi}) - E^{(s)}_{\alpha,\beta}(\ket{\phi})\right| 
\leq \frac{2\alpha\epsilon}{\alpha-1}.
\end{equation}

In the limit \( \alpha \to 1 \), the unified-entropy concentratable entanglements \( E^{(s)} \) converges to the von Neumann concentratable entanglements \(E^{(s)}(\rho):=\lim_{\alpha\to1}E^{(s)}_{\alpha,\beta}(\rho)\), which satisfies the continuity bound:
\begin{equation}\label{eq:contin_2}
\left|E^{(s)}(\ket{\psi}) - E^{(s)}(\ket{\phi})\right| 
\leq \epsilon \log_{2}(d-1) + h(\epsilon),
\end{equation}
where \( h(\epsilon) = -\epsilon \log_{2} \epsilon - (1 - \epsilon) \log(1 - \epsilon) \) is the binary entropy function, and \( d \) denotes the Hilbert space dimension associated with the states \( \psi \) and \( \phi \).  
\end{theorem}
\begin{proof}
First, we will prove Eq.~(\ref{eq:contin_1}). Let $\psi_{\chi}$ and $\phi_{\chi}$ be the reduced quantum states in the subsystems associated with \(\chi\). Then, we have
\begin{eqnarray}
|E^{(s)}_{\alpha,\beta}(\ket{\psi}) - E^{(s)}_{\alpha,\beta}(\ket{\phi})| 
&=&\nonumber \frac{1}{2^{|s|}}|\sum_{\chi\in\mathcal{P}(s)}S_{\alpha,\beta}(\psi_{\chi})-\sum_{\chi\in\mathcal{P}(s)}S_{\alpha,\beta}(\phi_{\chi})| 
\\\nonumber &=&   
\frac{1}{2^{|s|}}|\sum_{\chi\in\mathcal{P}(s)}(S_{\alpha,\beta}(\psi_{\chi})-S_{\alpha,\beta}(\phi_{\chi}))| 
\\\nonumber &\overset{(a)}{\leq}&   
\frac{1}{2^{|s|}}\sum_{\chi\in\mathcal{P}(s)}|S_{\alpha,\beta}(\psi_{\chi})-S_{\alpha,\beta}(\phi_{\chi})| 
\\\nonumber &\overset{(b)}{\leq}&   
\frac{1}{2^{|s|}}\sum_{\chi\in\mathcal{P}(s)} \frac{2\alpha}{\alpha-1}  D(\psi_{\chi},\phi_{\chi})
\\\nonumber &\overset{(c)}{\leq}&   
\frac{1}{2^{|s|}}\sum_{\chi\in\mathcal{P}(s)} \frac{2\alpha}{\alpha-1}  D(\psi,\phi)
\\ &\overset{}{\leq}&   
\frac{2\alpha\epsilon}{\alpha-1}, 
\end{eqnarray}    
where (a) follows from the triangle inequality, (b) follows from \(|S_{\alpha,\beta}(\psi_{\chi})-S_{\alpha,\beta}(\phi_{\chi})| \leq \frac{2\alpha}{\alpha-1}  D(\psi_{\chi},\phi_{\chi})\) holds for parameters \( \alpha > 1 \) and \( \beta \geq 1 \)~\cite{hu2006generalized}, (c) holds follows from the trace distance $D(\rho,\sigma)$ for two density operators is non-increasing under the partial trace operation.

Second, we will prove Eq.(\ref{eq:contin_2}). Let $d_{\chi}$ be the dimension of the systems associated with the reduced quantum states $\psi_{\chi}$ and $\phi_{\chi}$, the following chains of inequalities holds
\begin{widetext}
\begin{eqnarray}
|E^{(s)}(\ket{\psi}) - E^{(s)}(\ket{\phi})| &=&\nonumber \frac{1}{2^{|s|}}|\sum_{\chi\in\mathcal{P}(s)}S(\psi_{\chi})-\sum_{\chi\in\mathcal{P}(s)}S(\phi_{\chi})| 
\\\nonumber &=&   
\frac{1}{2^{|s|}}|\sum_{\chi\in\mathcal{P}(s)}(S(\psi_{\chi})-S(\phi_{\chi}))| 
\\\nonumber &\overset{}{\leq}&   
\frac{1}{2^{|s|}}\sum_{\chi\in\mathcal{P}(s)}|S(\psi_{\chi})-S(\phi_{\chi})| 
\\\nonumber &\overset{(a)}{\leq}&   
\frac{1}{2^{|s|}}\sum_{\chi\in\mathcal{P}(s)}\left[ D(\psi_{\chi},\phi_{\chi})\log_{2}(d_{\chi}-1)+h(D(\psi_{\chi},\phi_{\chi})) \right]
\\\nonumber &\overset{(b)}{\leq}&   
\frac{1}{2^{|s|}}\sum_{\chi\in\mathcal{P}(s)}
\left[ D(\psi,\phi)\log_{2}(d-1)+h(D(\psi,\phi)) \right]
\\\nonumber &=&   
D(\psi,\phi)\log_{2}(d-1)+h(D(\psi,\phi)) 
\\ &\overset{(b)}{\leq}&   
\epsilon\log_{2}(d-1)+h(\epsilon), 
\end{eqnarray}    
\end{widetext}
where (a) follows from the Fannes–Audenaert inequality $|S(\psi_{\chi})-S(\phi_{\chi})|\leq D(\psi_{\chi},\phi_{\chi})\log(d_{\chi}-1)+h(D(\psi_{\chi},\phi_{\chi}))$~\cite{fannes1973continuity,Nielsen2010,audenaert2007sharp-vonNUMANN}, (b) holds follows from the trace distance $D(\rho,\sigma)$ for two density operators is non-increasing under the partial trace operation and $d_{\chi}\leq d$, (c) holds because $h(x)$ is increasing function for $0\leq x\leq \frac{1}{2}$. 
\end{proof}

These continuity properties imply that small errors in estimating a quantum pure state will not lead to large fluctuations in the unified-entropy concentratable entanglements. This robustness is crucial in practical scenarios where one aims to measure multipartite entanglement in a quantum state.

\section{Ordering relations of unified-entropy concentratable entanglements}\label{sec:lowerbound}

In this section, we discuss the ordering relations of unified-entropy concentratable entanglements and how these relations can be estimated on near-term quantum devices.

We begin by formally introducing the definition of concentratable entanglements. For an \(n\)-qubit pure state \(\ket{\psi}\) and a subset \(s \subseteq [n]\) of the qubits, the concentratable entanglements \(C^{(s)}\) is defined as~\cite{beckey2021computable}:
\begin{equation}
C^{(s)}(\ket{\psi}) = 1 - \frac{1}{2^{|s|}} \sum_{\chi \in \mathcal{P}(s)} \mathrm{Tr}\, (\psi_\chi^2),
\end{equation}
where \(\mathcal{P}(s)\) denotes the power set of \(s\), and \(\psi_\chi\) is the reduced density matrix on the subsystem \(\chi\).
For a $n$-qubit mixed state \(\rho\), the concentratable entanglements is defined as
\begin{equation}
C^{(s)}(\rho) =\min\sum_ip_iC^{(s)}(\ket{\psi_i}) ,
\end{equation}
where the minimum is taken over all the pure state decomposition ${\{p_i,\psi_i\}}$ of $\rho$.

As previously discussed, this quantity corresponds to a special case of the unified-entropy concentratable entanglements evaluated at \(\alpha = 2\) and \(\beta = 1\).  

The R\'{e}nyi, unified-entropy, von Neumann, and concentratable entanglements satisfy the following (ordering) relation:

\begin{theorem}\label{the:lowerbound}
Let $\ket{\psi}$ be an arbitrary $n$-qudit pure state. 
For any orders $\alpha \geq \alpha' > 0$ and $\beta \geq 1$, the R\'{e}nyi concentratable entanglements $\mathcal{R}_{\alpha}^{(s)}$ and the unified-entropy concentratable entanglements $E^{(s)}_{\alpha',\beta}$ satisfy
\begin{equation}\label{eq:lowerbound-1}
   \mathcal{R}_{\alpha'}^{(s)}(\ket{\psi}) \geq \mathcal{R}_{\alpha}^{(s)}(\ket{\psi}),
\end{equation}
and
\begin{equation}\label{eq:lowerbound-2}
   E^{(s)}_{\alpha',\beta}(\ket{\psi}) \geq E^{(s)}_{\alpha,\beta}(\ket{\psi}).
\end{equation}
In particular, the von Neumann concentratable entanglements $E^{(s)}$ obeys
\begin{equation}\label{eq:lowerbound-3}
   E^{(s)}(\ket{\psi}) \geq \max \Bigl\{ \tfrac{1}{\ln 2}\, C^{(s)}(\ket{\psi}),\, 2 C^{(s)}(\ket{\psi}) - \tfrac{1}{2} \Bigr\}.
\end{equation}
Moreover, the R\'{e}nyi concentratable entanglements $\mathcal{R}_{\alpha}^{(s)}$ with $\alpha=2$ satisfies
\begin{equation}\label{eq:lowerbound-renyi-2}
   \mathcal{R}_{2}^{(s)}(\ket{\psi}) \geq \tfrac{1}{\ln 2}\, C^{(s)}(\ket{\psi}).
\end{equation}
\end{theorem}
\begin{proof}
We first prove that Eq.~\eqref{eq:lowerbound-1} holds. 
Recall that 
\(\mathcal{R}^{(s)}_{\alpha}(\ket{\psi}) = \tfrac{1}{2^{|s|}} \sum_{\chi \in \mathcal{P}(s)} \mathcal{R}_{\alpha}(\psi_{\chi})\),
where \(\psi_{\chi}\) denotes the joint reduced state of the subsystems corresponding to the elements in \(\chi\). 
It is then sufficient to show that, for \(\alpha \geq \alpha' > 0\),
\begin{equation}\label{eq:renyi-order}
R_{\alpha'}(\psi_{\chi}) \geq R_{\alpha}(\psi_{\chi}),
\end{equation}
which has already been established in Ref.~\cite{muller2013quantum}. 
Therefore, Eq.~\eqref{eq:lowerbound-1} follows immediately.

We now turn to Eq.~(\ref{eq:lowerbound-2}). 
Analogous to the proof of Eq.~(\ref{eq:lowerbound-1}), it suffices to show that, for \(\alpha \geq \alpha' > 0\) and \(\beta\geq1\),
\begin{equation}\label{eq:Tsallis-order} 
 S_{\alpha,\beta}(\psi_{\chi}) \geq  S_{\alpha,\beta}(\psi_{\chi}).
\end{equation}
A detailed proof of Eq.~(\ref{eq:Tsallis-order}) is given in Appendix~\ref{app:proof-Tsallis-order}.

Then, we will prove Eq.~\eqref{eq:lowerbound-3}. 
We first establish that
\begin{equation}\label{eq:vN-linear-order}
   E^{(s)}(\ket{\psi}) \geq \frac{1}{\ln 2}\, C^{(s)}(\ket{\psi}).
\end{equation}
To prove Eq.~(\ref{eq:vN-linear-order}), it suffices to show that
\begin{equation}
S(\psi_\chi) \geq \tfrac{1}{\ln 2}\, S_{\mathrm{lin}}(\psi_\chi)
\;\;\Longrightarrow\;\;
-\sum_i \lambda_i \ln \lambda_i \;\geq\; 1 - \sum_i \lambda_i^2,
\end{equation}
where \(\{\lambda_i\}\) are the eigenvalues of \(\psi_\chi\). 
This follows directly from the elementary inequality \(\ln x \leq x-1\) valid for all \(x>0\).

Next, we show that 
\begin{equation}\label{eq:vN-linear-odered-2}
   E^{(s)}(\ket{\psi}) \geq 2 C^{(s)}(\ket{\psi}) - \tfrac{1}{2}.
\end{equation}
First, we establish that for any reduced state \(\psi_\chi\) of \(\ket{\psi}\), the following inequality holds: \(S(\psi_\chi) \geq 2(1 - \Tr \psi_\chi^2) - \frac{1}{2}\). The proof is outlined as follows:  
\begin{eqnarray}
S(\psi_\chi) &=&\nonumber  \sum_{k=1}^{\infty}\frac{1}{k}\Tr(\rho(I-\psi_\chi)^k)
\\&\geq&\nonumber
  \Tr(\rho(I-\psi_\chi))+\frac{1}{2}\Tr(\rho(I-\psi_\chi)^k)
\\&\geq&\nonumber
  \frac{3}{2}-2\Tr\psi_\chi^2= 2(1-\Tr\psi_\chi^2)-\frac{1}{2},
\end{eqnarray}
where the first equality follows from the Taylor series expansion of the von Neumann entropy and the first inequality holds because \(\Tr(\rho(I - \psi_\chi)^k) \geq 0\) for any \(k \in \mathbb{N}^+\).

Second, we have that 
\begin{eqnarray}
    E^{(s)}(\psi) &=& \nonumber \frac{1}{2^{|s|}} \sum_{\chi\in\mathcal{P}(s)}S(\psi_\chi)
\\&\geq&\nonumber
   \frac{1}{2^{|s|}} \sum_{\chi\in\mathcal{P}(s)}[2(1-\Tr\psi_\chi^2)-\frac{1}{2}] 
\\&=&\nonumber
  2\left(1- \frac{1}{2^{|s|}}\sum_{\chi\in\mathcal{P}(s)}\Tr\psi_\chi^2\right)-\frac{1}{2}   
\\&=&\nonumber
  2C^{(s)}(\ket{\psi})-\frac{1}{2}.    
\end{eqnarray}

Finally, we establish Eq.~(\ref{eq:lowerbound-renyi-2}), which follows from the inquality 
\(\ln x\leq x-1\) valid for all \(x>0\).

This completes the proof.
\end{proof}

\begin{remark}
Theorem~\ref{the:lowerbound} characterizes the ordering relations among unified-entropy concentratable entanglements under different parameter choices. For instance, in Eq.~(\ref{eq:lowerbound-1}), setting $\alpha'=1$ and $\alpha=2$ yields a relation between the von Neumann concentratable entanglement \(E^{(s)}\) and the R\'{e}nyi concentratable entanglement \(\mathcal{R}_{\alpha}^{(s)}\) for \(\alpha=2\):
\begin{equation}\label{eq:eq41}
E^{(s)}(\ket{\psi}) \geq \frac{1}{\ln 2}\mathcal{R}_2^{(s)}(\ket{\psi}) > \mathcal{R}_2^{(s)}(\ket{\psi}).
\end{equation}
Similarly, when $\beta=1$, Eq.~(\ref{eq:lowerbound-2}) reduces to the case of Tsallis concentratable entanglement. i.e., 
\begin{equation}\label{eq:lowerbound-4-pure}
   \mathcal{T}_{\alpha'}^{(s)}(\ket{\psi}) \geq \mathcal{T}_{\alpha}^{(s)}(\ket{\psi}). 
\end{equation}
Furthermore, by choosing $\alpha'=2$ and $\alpha=3$, we obtain a relation between the concentratable entanglement \(C^{(s)}\) and the Tsallis concentratable entanglement \(\mathcal{T}_{\alpha}^{(s)}\) for \(\alpha=3\):
\begin{equation}\label{eq:45}
C^{(s)}(\ket{\psi}) \geq \mathcal{T}_3^{(s)}(\ket{\psi}).
\end{equation}
Finally, combining Eq.~(\ref{eq:lowerbound-renyi-2}) with Eq.~(\ref{eq:eq41}) and Eq.~(\ref{eq:45}), we further obtain
\begin{equation}\label{eq:lowerbound-renyi-2'}
E^{(s)}(\ket{\psi})>\mathcal{R}_{2}^{(s)}(\ket{\psi}) > C^{(s)}(\ket{\psi}) \geq \mathcal{T}_3^{(s)}(\ket{\psi}).
\end{equation}
These relations are further illustrated in Section~\ref{sec:examples} through explicit examples. 
\end{remark}

Theorem \ref{the:lowerbound} can be extended to the case of mixed states as well:
\begin{corollary}\label{coro:vn-mixed-lowerbound}
Let $\rho$ be an arbitrary $n$-qudit mixed state. 
For any orders $\alpha\geq\alpha'>0$ and $\beta\geq1$ the R\'{e}nyi concentratable entanglements $\mathcal{R}_{\alpha}^{(s)}$ and the unified-entropy concentratable entanglements $E^{(s)}_{\alpha',\beta}$ satisfy 
\begin{equation}\label{eq:lowerbound-1-1}
   \mathcal{R}_{\alpha'}^{(s)}(\rho) \geq \mathcal{R}_{\alpha}^{(s)}(\rho),
\end{equation}
and
\begin{equation}\label{eq:lowerbound-2-1}
E^{(s)}_{\alpha',\beta}(\rho) \geq E^{(s)}_{\alpha,\beta}(\rho) 
\end{equation}    
In particular, the von Neumann concentratable entanglements $E^{(s)}$ obeys
\begin{equation}\label{eq:lowerbound-3-1}
   E^{(s)}(\rho) \geq \max \Bigl\{ \frac{1}{\ln2} C^{(s)}(\rho),\, 2 C^{(s)}(\rho)-\tfrac{1}{2} \Bigr\}.
\end{equation}
\end{corollary}

\begin{proof}
Here, we prove only the second part of Eq.~(\ref{eq:lowerbound-3-1}). The proofs of the remaining inequalities follow in exactly the same way.
Let \(\rho=\sum_ip_i\psi_i\) be the optimal decomposition of \(E^{(s)}(\rho)\). Then, we have
\begin{eqnarray}
    E^{(s)}(\rho) &=& \nonumber \sum_ip_iE^{(s)}(\ket{\psi_i})
\\&\geq&\nonumber
    \sum_{i}p_i[2C^{(s)}(\ket{\psi_i})-\frac{1}{2}] 
\\&=&\nonumber
   2\sum_{i}p_iC^{(s)}(\ket{\psi_i})-\frac{1}{2}
\\&\geq&\nonumber
  2C^{(s)}(\rho)-\frac{1}{2},
\end{eqnarray}
where the first equality follows from Eq.~(\ref{eq:lowerbound-3}).
This completes the proof.
\end{proof}

By Eq.~(\ref{eq:lowerbound-2-1}), we further obtain that the Tsallis concentratable entanglements $\mathcal{T}_{\alpha}^{(s)}$ satisfis 
\begin{equation}\label{eq:lowerbound-4}
   \mathcal{T}_{\alpha'}^{(s)}(\rho) \geq \mathcal{T}_{\alpha}^{(s)}(\rho),
\end{equation}
for any orders $\alpha\geq\alpha'>0$.

Let \(Z = \{0,1\}^n\) be the set of all such bitstrings and \(Z_0(s)\) is the set of all bitstrings with 0's on all indices in \(s\). In Ref.~\cite{beckey2021computable}, Beckey $et$ $al.$ show that the the concentratable entanglements of a \(n\)-qubit pure state can be computed from the outcomes of the \(n\)-qubit parallelized SWAP test as
\begin{equation}
C_s(\ket{\psi}) = 1 - \sum_{z \in Z_0(s)} p(z),
\end{equation}
where \(p(z)\) is the probability of measuring the \(z\)-bitstring on the \(n\)-control qubits. Thus, through the \(n\)-qubit parallelized SWAP test, we can efficiently estimate a lower bound of the von Neumann concentratable entanglements and the R\'{e}nyi concentratable entanglement $\mathcal{R}_{\alpha}^{(s)}$ with $\alpha=2$ on real quantum devices (e.g., using quantum optical Fredkin gates~\cite{Milburn1989-PhysRevLett.62.2124}). 

In Ref.\cite{Beckey2023_PhysRevA.107.062425}, Beckey $et$ $al.$ demonstrated how to construct lower bounds on the concentratable entanglement of mixed states that can be estimated using only Bell-basis measurement data. By combining their method with Corollary~\ref{coro:vn-mixed-lowerbound}, we further show that a lower bound on the von Neumann concentratable entanglement and the R\'{e}nyi concentratable entanglement $\mathcal{R}_{\alpha}^{(s)}$ with $\alpha=2$ of a mixed quantum state $\rho$ can also be estimated from Bell-basis measurement data.

\section{Examples}\label{sec:examples}
In this section, we present several examples of von Neumann concentratable entanglements, Tsallis concentratable entanglements, and R\'{e}nyi concentratable entanglements for \(n\)-qubit GHZ and W states. We then investigate the unified-entropy concentratable entanglements in a four-partite star quantum network. 
In the following examples, we use the following four specific instances of the unified-entropy concentratable entanglements as benchmarks for our study: the von Neumann concentratable entanglements \(E^{(s)}\), the R\'{e}nyi concentratable entanglements \(\mathcal{R}_{\alpha}^{(s)}\) for \(\alpha=2\), the Tsallis concentratable entanglements \(\mathcal{T}_{\alpha}^{(s)}\) for \(\alpha=3\), and the concentratable entanglements \(C^{(s)}\).

\subsection{\texorpdfstring{$n$-qubit GHZ states and W states}{n-qubit GHZ states and W states}}
Under the framework of reversible LOCC, GHZ and W states remain maximally entangled within their respective classes. However, from the perspective of various pure-state entanglement measures, GHZ states are generally considered to have a higher degree of entanglement than W states. In contrast, W states exhibit greater robustness, as the remaining system can still retain entanglement even if some qubits are measured or lost. In this section, we focus on these two types of entangled states due to their significant applications in quantum computing~\cite{Horodecki2009}. 

Firstly, we will give some examples for von Neumann concentratable entanglements \(E^{(s)}\) for $n$-qubit GHZ states and W states. 

For \( n \geq 3 \), the \( n \)-qubit GHZ and W states, expressed in the computational basis, are given by  

\begin{equation}  
\ket{\text{GHZ}_n} = \frac{1}{\sqrt{2}} \left( \ket{0}^{\otimes n} + \ket{1}^{\otimes n} \right),  
\end{equation}  

\begin{equation}  
\ket{W_n} = \frac{1}{\sqrt{n}} \sum_{i=1}^{n} \ket{0 \dots 1_i \dots 0}.  
\end{equation}  

In Fig.~\ref{fig:GHZvsW-von-2}, we plot the difference \(\Delta = E^{(s)}(\ket{\text{GHZ}}) - E^{(s)}(\ket{W})\) for various subsystem sizes \(|s|\). 
Meanwhile, for \(s = [n]\), we have 
\[
E^{(s)}(\ket{\text{GHZ}}) = 1 - \frac{1}{2^{n-1}}
\] 
and  
\[
E^{(s)}(\ket{W}) = \frac{1}{2^n} \sum_{k=0}^{n} \binom{n}{k} h\left(\frac{k}{n}\right),
\]
where \(h\left(\frac{k}{n}\right) = -\frac{k}{n} \log \frac{k}{n} - \left(\frac{n-k}{n}\right) \log \left(\frac{n-k}{n}\right)\). As shown in Fig.~\ref{fig:GHZvsW-von-1}, both \(E^{(s)}(\ket{\text{GHZ}})\) and \(E^{(s)}(\ket{W})\) increase slowly and gradually approach 1 as \(n\) increases. 
The results show that for large \(n\), the GHZ and W states can be effectively distinguished using small subsystems, i.e., when \(|s|\) is small. In both cases, we observe that \(E^{(s)}(\ket{\text{GHZ}}) > E^{(s)}(\ket{W})\), indicating that \(E^{(s)}\) is an effective measure for distinguishing GHZ states from W states.

\begin{figure}
    \centering
    \includegraphics[width=0.95\linewidth]{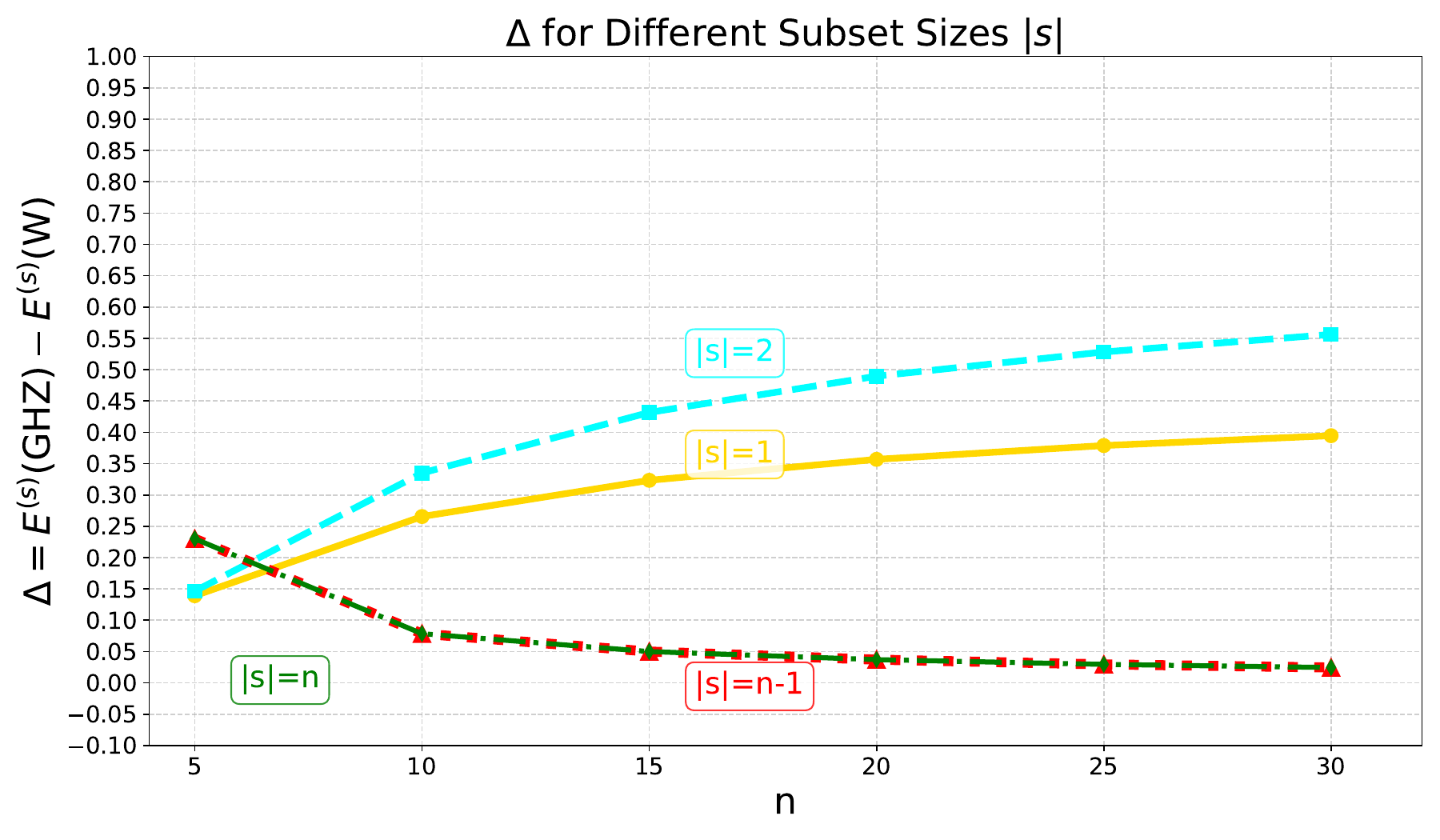}
    \caption{The difference \(\Delta=E^{(s)}(\ket{\text{GHZ}}) - E^{(s)}(\ket{W})\) for various subsystem sizes \(|s|\).}
    \label{fig:GHZvsW-von-2}
\end{figure}

\begin{figure}
    \centering
    \includegraphics[width=0.99\linewidth]{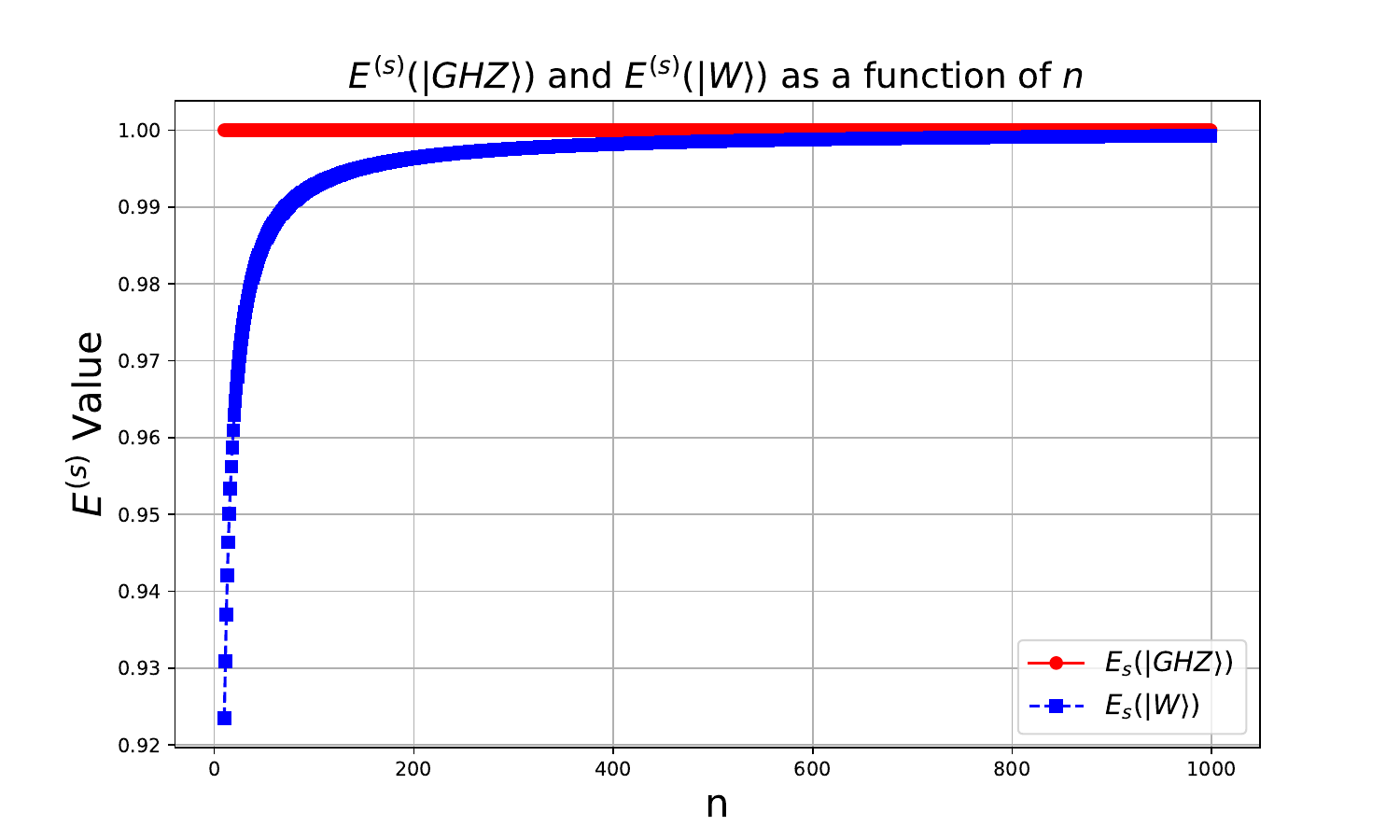}
    \caption{The von Neumann concentratable entanglements serves as an effective measure for distinguishing GHZ states from W states.}
    \label{fig:GHZvsW-von-1}
\end{figure}
\begin{figure}
    \centering
    \includegraphics[width=0.99\linewidth]{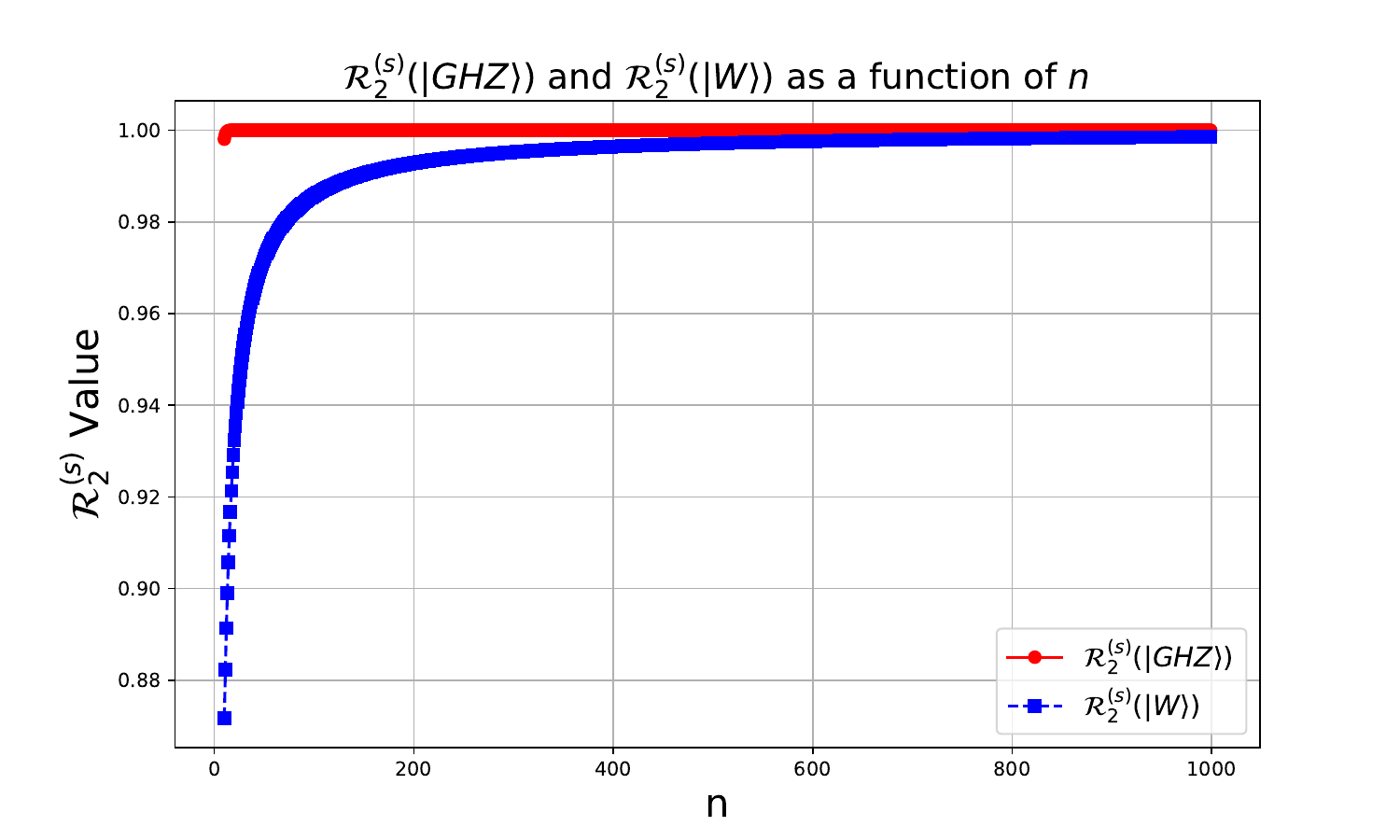}
    \caption{The R\'{e}nyi concentratable entanglements serves as an effective measure for distinguishing GHZ states from W states.}
    \label{fig:GHZvsW-renyi-1}
\end{figure}
\begin{figure}
    \centering
    \includegraphics[width=0.99\linewidth]{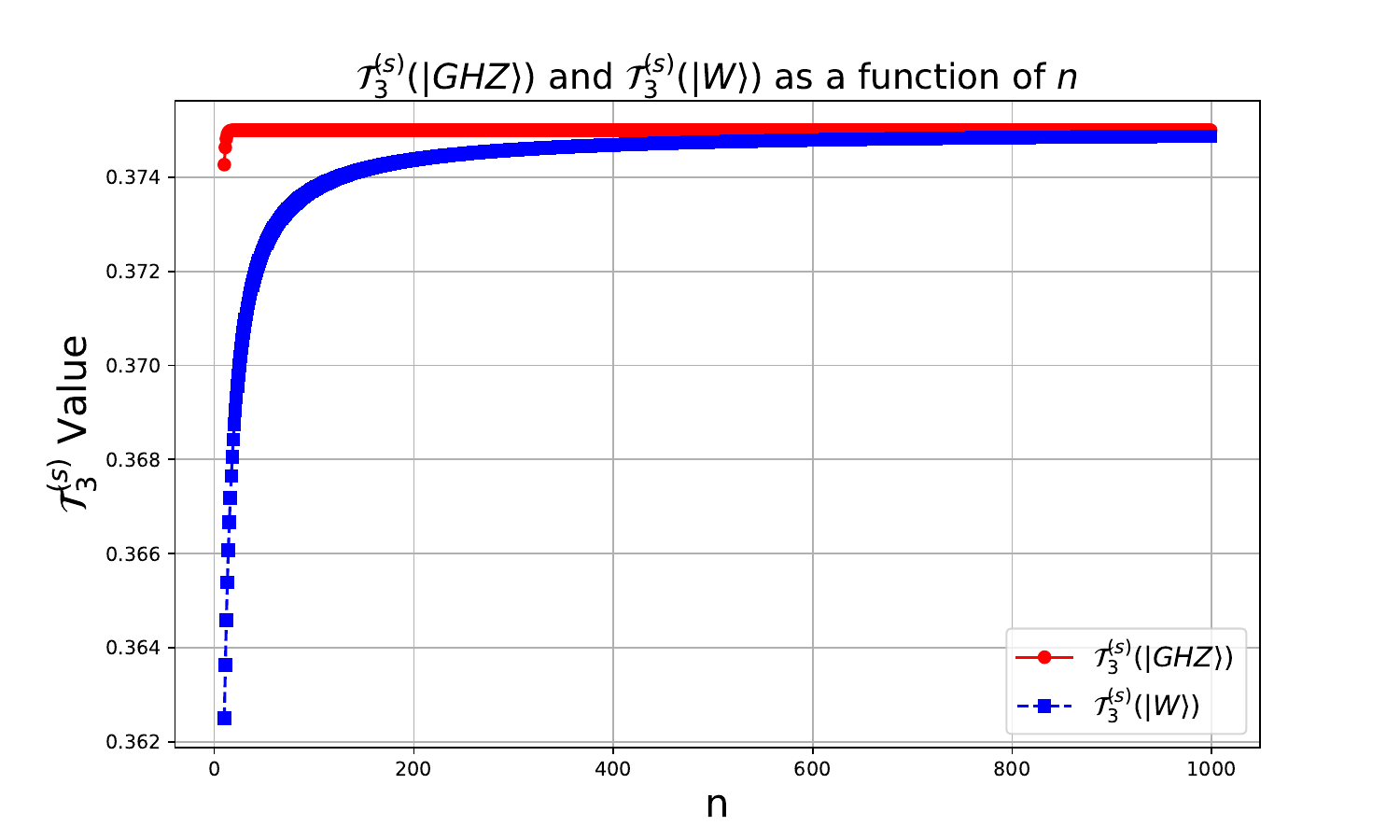}
    \caption{The Tsallis concentratable entanglements serves as an effective measure for distinguishing GHZ states from W states.}
    \label{fig:GHZvsW-tsallis_3}
\end{figure}
\begin{figure}
    \centering
    \includegraphics[width=0.99\linewidth]{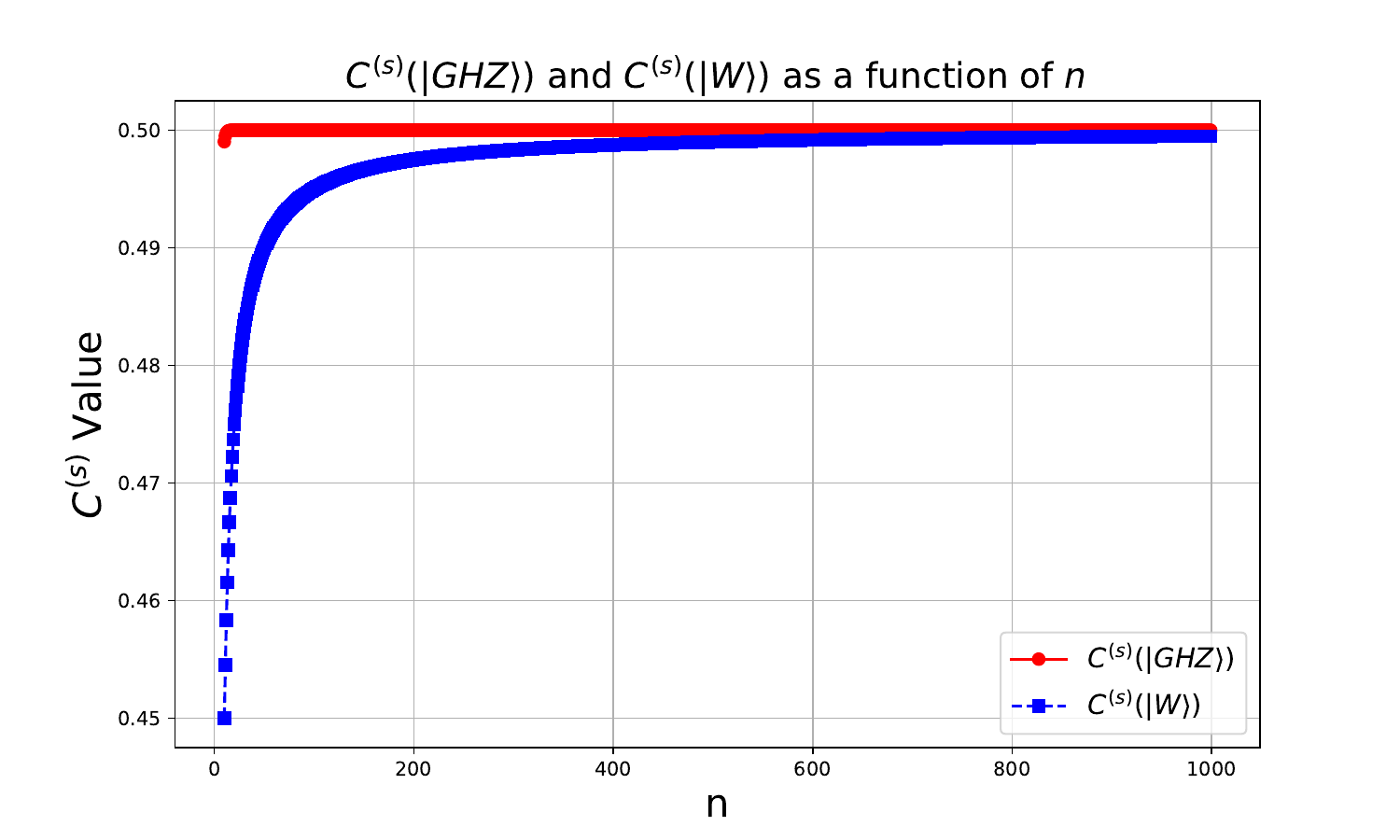}
    \caption{The concentratable entanglements serves as an effective measure for distinguishing GHZ states from W states.}
    \label{fig:GHZvsW-linear}
\end{figure}
Similarly, we also compare the behaviors of the R\'{e}nyi concentratable entanglements (at \(\alpha=2\)), the Tsallis concentratable entanglements (at \(\alpha=3\)), and the concentratable entanglements as functions of \(n\) for GHZ and W states.  
The explicit expressions for these quantities are given by
\begin{widetext}
\[
\mathcal{R}_2^{(s)}(\ket{\text{GHZ}}) = 1 - \frac{1}{2^{n-1}}, \quad
\mathcal{R}_2^{(s)}(\ket{W}) = -\frac{1}{2^n} \sum_{k=0}^{n} \binom{n}{k} \log_2\left(\left(\frac{k}{n}\right)^2 + \left(\frac{n-k}{n}\right)^2\right),
\]
\[
\mathcal{T}_3^{(s)}(\ket{\text{GHZ}}) = \frac{3}{8} \frac{2^n-2}{2^n}, \quad
\mathcal{T}_3^{(s)}(\ket{W}) = \frac{3n-1}{8n},
\]
\[
C^{(s)}(\ket{\text{GHZ}}) = \frac{1}{2} \frac{2^n-2}{2^n}, \quad
C^{(s)}(\ket{W}) = \frac{n-1}{2n}.
\]  
\end{widetext}
By analyzing the trends shown in Figs.~\ref{fig:GHZvsW-renyi-1}--\ref{fig:GHZvsW-linear}, we were surprised to find that all four multipartite entanglement measures—including the von Neumann concentratable entanglements—not only effectively distinguish GHZ states from W states, but also exhibit remarkably similar behavior as \(n\) increases.
\subsection{A four-partite star quantum network}\label{sec:starnetwork}
Consider following quantum state 
\begin{equation}
    \ket{\psi}=
(\cos{\theta}\ket{00}+\sin{\theta}\ket{11})^{\otimes 3}.
\end{equation}
It can be seen that this quantum state is composed of the tensor product of three EPR-like states. If the first register of each EPR-like state is regarded as a subsystem of the first system, then the quantum state \(\ket{\psi}\) can be viewed as a four-partite pure state residing in the \(8 \otimes 2 \otimes 2 \otimes 2\)-dimensional Hilbert space \(\mathcal{H}_A\otimes\mathcal{H}_B\otimes\mathcal{H}_C\otimes\mathcal{H}_D\). Thus, we have

\begin{eqnarray}
\ket{\psi}
&=&
a^3 \ket{0000} + a^2 b \ket{1001} + a^2 b \ket{2010} + a b^2 \ket{3011} \nonumber 
\\\nonumber&& 
+ a^2 b \ket{4100} + a b^2 \ket{5101} + a b^2 \ket{6110} + b^3 \ket{7111},
\end{eqnarray}
with \(a=\cos{\theta}, b=\sin{\theta}\). For \(s = [n]\), the unified-entropy concentratable entanglements \(E_{\alpha,\beta}^{(s)}(\ket{\psi})\) can be expressed as  
\[
E_{\alpha,\beta}^{(s)}(\ket{\psi}) = \frac{1}{8} \left(3S_{\alpha,\beta}(\psi_{A}) + 3S_{\alpha,\beta}(\psi_{B}) + S_{\alpha,\beta}(\psi_{AB}) \right),
\]  
where \(\psi_{A}\), \(\psi_{B}\), and \(\psi_{AB}\) denote the reduced density matrices of the subsystems corresponding to \(\mathcal{H}_A\), \(\mathcal{H}_B\), and \(\mathcal{H}_{AB}\), respectively. 

We conduct a comparative study of entanglement behavior in the four-partite star quantum network using four entanglement measures: the von Neumann concentratable entanglements \(E^{(s)}\), the R\'{e}nyi concentratable entanglements \(\mathcal{R}_{2}^{(s)}\), the Tsallis concentratable entanglements \(\mathcal{T}_3^{(s)}\), and the linear concentratable entanglements \(C^{(s)}\). As shown in Fig.~\ref{fig:psi_vonNuemann}, compared to \(\mathcal{T}_3^{(s)}\) and \(C^{(s)}\), both \(E^{(s)}\) and \(\mathcal{R}_{2}^{(s)}\) exhibit superior sensitivity in distinguishing the entanglement variations across different values of the parameter \(\theta\) in the state \(\ket{\psi(\theta)}\).

\begin{figure}
\centering
    \includegraphics[width=0.94\linewidth]{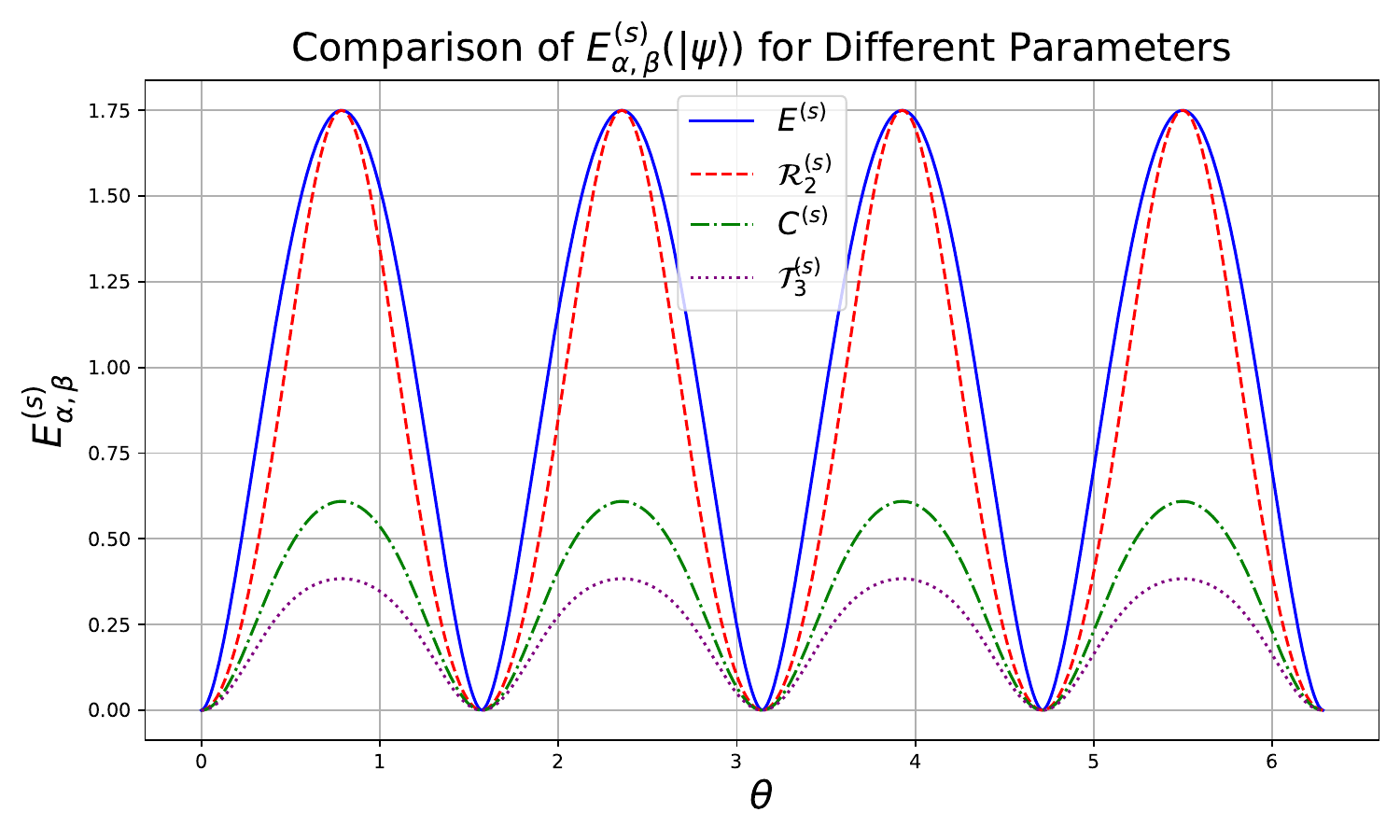}
     \caption{Comparison of unified-entropy concentratable entanglements for von Nuemann concentratable entanglemnt \(E^{(s)}\), R\'{e}nyi concentratable entanglements \(\mathcal{R}_{2}^{(s)}\), Tsallis concentratable entanglements \(\mathcal{T}_3^{(s)}\), and concentratable entanglements \(C^{(s)}\).}
    \label{fig:psi_vonNuemann}
\end{figure}
\subsection{Four-qubit Dicke states}\label{Dicke}
Dicke states provide a rich variety of structurally complex states among many particles and hold great promise for a wide range of applications in quantum information~\cite{Chiuri2012-PhysRevLett.109.173604,Lucke2014-PhysRevLett.112.155304,Adam2021-PhysRevA.104.022426,Chen2020-PhysRevA.101.012308}. In the following, we conduct a comparative study of the four measures—the von Neumann concentratable entanglement $E^{(s)}$, the R\'{e}nyi concentratable entanglement $\mathcal{R}_{2}^{(s)}$, the Tsallis concentratable entanglement $\mathcal{T}_{3}^{(s)}$, and the linear concentratable entanglement $C^{(s)}$—with respect to their ability to distinguish four-qubit Dicke states.

For an $n$-qubit system, the Dicke state with $k$ excitations is defined as the 
completely symmetric superposition of all computational basis states with exactly 
$k$ qubits in the excited state $\ket{1}$ and the remaining $n-k$ qubits in the ground 
state $\ket{0}$. Explicitly, it can be written as~\cite{Chiuri2012-PhysRevLett.109.173604}
\begin{equation}\label{Dickestates}
    \ket{D^{(n)}_{k}}
    = \binom{n}{k}^{-1/2} 
    \sum_{\pi \in S_n} U_{\pi} \Big( \ket{1}^{\otimes k} \otimes \ket{0}^{\otimes (n-k)} \Big),
\end{equation}
where $S_n$ denotes the symmetric group of degree $n$ and $U_{\pi}$ is the permutation 
operator corresponding to $\pi \in S_n$. The normalization factor 
$\binom{n}{k}^{-1/2}$ ensures that $\ket{D^{(n)}_{k}}$ is a unit vector.
For the case of four qubits $(n=4)$, the Dicke states 
$\{\ket{D^{(4)}_k}\}_{k=0}^4$ take the following explicit forms:
\begin{align}
    \ket{D^{(4)}_0} &= \ket{0000}, \\[6pt]
    \ket{D^{(4)}_1} &= \tfrac{1}{2}\big( \ket{1000} + \ket{0100} + \ket{0010} + \ket{0001} \big), \\[6pt]
    \ket{D^{(4)}_2} &= \tfrac{1}{\sqrt{6}}\big( \ket{1100} + \ket{1010} + \ket{1001} 
                                   + \ket{0110} + \ket{0101} + \ket{0011} \big), \\[6pt]
    \ket{D^{(4)}_3} &= \tfrac{1}{2}\big( \ket{1110} + \ket{1101} + \ket{1011} + \ket{0111} \big), \\[6pt]
    \ket{D^{(4)}_4} &= \ket{1111}.
\end{align}
These states correspond respectively to zero, one, two, three, and four excitations 
in the computational basis. For \(s = [n]\), the unified-entropy concentratable entanglements \(E_{\alpha,\beta}^{(s)}(\ket{D^{(4)}_k})\) can be expressed as  
\[
E_{\alpha,\beta}^{(s)}(\ket{D^{(4)}_k}) = \frac{1}{8} \left(4S_{\alpha,\beta}(\psi_{A}) + 3S_{\alpha,\beta}(\psi_{AB}) \right),
\]  
where \(\psi_{A}\) and \(\psi_{AB}\) denote the reduced density matrices of the subsystems corresponding to \(\mathcal{H}_A\) and \(\mathcal{H}_{AB}\), respectively. 
As shown in Fig.~\ref{fig:Dicke-states}, compared with $\mathcal{T}_3^{(s)}$ and $C^{(s)}$, the values of $E^{(s)}$ and $\mathcal{R}_{2}^{(s)}$ are larger and provide better distinction among four-qubit Dicke states with different $k$.   

From both Fig.~\ref{fig:psi_vonNuemann} and Fig.~\ref{fig:Dicke-states}, the ordering relations among the four multipartite entanglement measures $E^{(s)}$, $\mathcal{R}_{2}^{(s)}$, $\mathcal{T}_3^{(s)}$, and $C^{(s)}$ can be clearly observed, which further provides a concrete illustration of Theorem~\ref{the:lowerbound}.

\begin{figure}
\centering
    \includegraphics[width=0.94\linewidth]{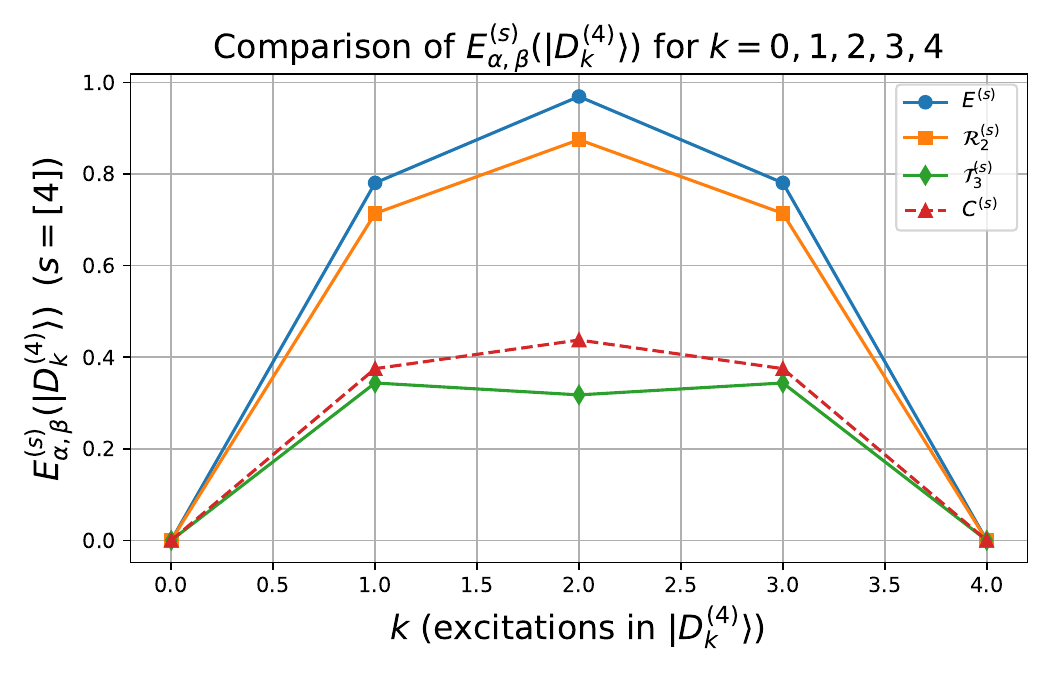}
     \caption{Comparison of different unified-entropy concentratable entanglements for four-qubit Dicke states.}
    \label{fig:Dicke-states}
\end{figure}

\section{Conclusion}\label{sec:conclusion}
In this paper, we introduced a family of two-parameter multipartite entanglement measures, termed unified-entropy concentratable entanglements. This class of measures not only provides a unified framework for detecting, quantifying, and characterizing multipartite entanglement, but also encompasses a wide range of both bipartite and multipartite entanglement measures. Specifically, it includes well-known bipartite measures such as the entanglement of formation, R\'{e}nyi entanglement, and concurrence, as well as multipartite measures such as the original concentratable entanglements and the entanglement measures proposed in Ref.~\cite{Carvalho2004-PhysRevLett.93.230501}. We prove that for any quantum state and for all \(\alpha>0, \beta\geq0\), the unified-entropy concentratable entanglements
$E^{(s)}_{\alpha,\beta}$ is monotonically non-increasing under multipartite separable operations. A direct corollary of this result is that $E^{(s)}_{\alpha,\beta}$ is also monotonically non-increasing under LOCC operations. Furthermore, we establish several desirable properties of this entanglement measure, including subadditivity and continuity.

In addition, we specifically focus on the von Neumann concentratable entanglements $E^{(s)}$. We found that $E^{(s)}$ not only constitutes a well-defined family of entanglement measures, but also exhibits several advantages over the concentratable entanglements $C^{(s)}$. In particular, as shown in Eq.~(\ref{eq:vn-re-add-1}) and Eq.~(\ref{eq:vn-ren-add-2}), $E^{(s)}$ satisfies additivity for biseparable pure states, whereas $C^{(s)}$ only satisfies a form of "pseudo-additivity" (as shown in Eq.~(\ref{eq:qsi-additive})). Moreover, we establish a series of ordering relations among the unified-entropy concentratable entanglements. For example, we prove that
\begin{equation}
E^{(s)}(\rho)>\mathcal{R}_{2}^{(s)}(\rho) > C^{(s)}(\rho) > \mathcal{T}_3^{(s)}(\rho),
\end{equation}
holds for any quantum state \(\rho\).
Since $C^{(s)}$ can be efficiently estimated from Bell-basis measurement data on near-term quantum devices, this result provides a practical means to estimate lower bounds for both $E^{(s)}$ and $\mathcal{R}_{2}^{(s)}$, as well as upper bounds for the Tsallis concentratable entanglements such as $\mathcal{T}_3^{(s)}$.


Furthermore, taking the von Neumann concentratable entanglements $E^{(s)}$, the R\'{e}nyi concentratable entanglements $\mathcal{R}_{2}^{(s)}$, the Tsallis concentratable entanglements $\mathcal{T}_{3}^{(s)}$, and the concentratable entanglements $C^{(s)}$ as benchmarks, we investigated their ability to distinguish between $n$-qubit GHZ and W states. Our results indicate that all four families of measures exhibit comparable performance in distinguishing these two classes of multipartite entangled states. In addition, we explored their performance in characterizing entanglement variations in a four-partite star quantum network $\ket{\psi(\theta)}$ and four-qubit Dicke states \(\ket{D^{(4)}_{k}}\). Compared to $\mathcal{T}_3^{(s)}$ and $C^{(s)}$, both $E^{(s)}$ and $\mathcal{R}_{2}^{(s)}$ demonstrate superior sensitivity in detecting entanglement changes as the parameter $\theta$ varies in the quantum state $\ket{\psi(\theta)}$. 
This further demonstrates that the multi-parameter entanglement measure $E^{(s)}_{\alpha,\beta}$ possesses potential advantages over the original concentratable entanglements $C^{(s)}$, particularly in tasks such as entanglement detection.

These findings not only expand the landscape of multipartite entanglement measures but also underscore the utility of characterizing multipartite entanglement through bipartite quantities. We hope that our results provide valuable tools for future studies on the structure and quantification of multipartite entanglement. In particular, promising directions for future work include the development of efficient quantum algorithms for estimating unified-entropy concentratable entanglements with arbitrary parameters, as well as deriving analytical formulas for special classes of quantum states, such as graph states~\cite{Cullen2022-PhysRevA.106.042411}.
\begin{acknowledgments}
We sincerely thank the two anonymous reviewers for their valuable feedback, which has contributed to improving the quality of our paper.
We also thank Yu Guo, Xue Yang, Zhengjun Xi and Jianwei Xu for helpful discussions.
This work is supported by the National Natural Science Foundation of China (Grant No. 62001274, Grant No. 12271325, Grant No. 62171266, Grant No. 12301590) and the Natural Science Foundation of the Higher Education Institutions of Jiangsu Province of China (No. 1020240975) and the Fundamental Research Funds for the Central Universities (GK202501008).
\end{acknowledgments}
%
\appendix
\section{Schur concavity of unified entropy}\label{app:schur-concave}
We now demonstrate that the unified entropy is Schur-concave for parameters $\alpha > 0$ and $\beta > 0$.
Let $\rho = \sum_i \lambda_i \ketbra{\psi_i}{\psi_i}$ be a quantum state, and define
$$
h(\lambda) := \frac{1}{(1-\alpha)\beta} \left[ \left( \sum_i \lambda_i^\alpha \right)^\beta - 1 \right].
$$
To establish the Schur-concavity of the unified entropy, it is sufficient to verify that for all $i,j$~\cite{bhatia2013matrix},
\begin{equation}
\Gamma_{ij} :=
(\lambda_i - \lambda_j) \left( \frac{\partial h}{\partial \lambda_i} - \frac{\partial h}{\partial \lambda_j} \right) \leq 0.
\end{equation}
Computing the partial derivatives, we obtain
\begin{equation}
\frac{\partial h}{\partial \lambda_i}=
\frac{\alpha}{1-\alpha}s^{\beta-1} \lambda_i^{\alpha-1},
\end{equation}
where $s = \sum_i \lambda_i^\alpha$.
Substituting into the expression for $\Gamma_{ij}$, we find
\begin{equation}
\Gamma_{ij}=
(\lambda_i - \lambda_j) \frac{\alpha}{1-\alpha} s^{\beta-1} \bigl(\lambda_i^{\alpha-1} - \lambda_j^{\alpha-1}\bigr).
\end{equation}
Since $s^{\beta-1}>0$, the sign of $\Gamma_{ij}$ is determined by
\begin{equation}
(\lambda_i - \lambda_j) \bigl(\lambda_i^{\alpha-1}-\lambda_j^{\alpha-1}\bigr)\frac{\alpha}{1-\alpha}.
\end{equation}
For $\alpha > 1$, we have $\tfrac{\alpha}{1-\alpha} < 0$, while $(\lambda_i - \lambda_j)(\lambda_i^{\alpha-1}-\lambda_j^{\alpha-1}) \geq 0$, which ensures $\Gamma_{ij} \leq 0$.
For $0 < \alpha < 1$, we have $\tfrac{\alpha}{1-\alpha} > 0$, while $(\lambda_i - \lambda_j)(\lambda_i^{\alpha-1}-\lambda_j^{\alpha-1}) \leq 0$, which again gives $\Gamma_{ij} \leq 0$.
For $\alpha = 1$ or $\beta=0$, $h(\lambda)$ reduces to the Shannon entropy or R\'{e}nyi entropy, both of which are well known to be Schur-concave.
Therefore, the unified entropy is Schur-concave for all $\alpha > 0$ and $\beta \geq 0$.

\section{Proofs of Eq.~(\ref{eq:Tsallis-order})}
\label{app:proof-Tsallis-order}
Let \(\{\lambda_i\}\) be the eigenvalues of the quantum state \(\psi_\chi\), and define the function 
\(f(\alpha) := \mathrm{Tr}(\psi_\chi^\alpha) = \sum_i \lambda_i^\alpha\) for \(\alpha > 0\). 
It is straightforward to compute
\begin{equation}
f'(\alpha) = \sum_i \lambda_i^\alpha \ln \lambda_i \leq 0, 
\qquad
f''(\alpha) = \sum_i \lambda_i^\alpha (\ln \lambda_i)^2 \geq 0.
\end{equation}
Therefore, \(f(\alpha)\) is monotonically nonincreasing and convex, with \(f(1) = 1\).
For any \(\alpha>0, \alpha \neq 1, \beta\geq1\), we define the function
\begin{equation}
g_\beta(\alpha) := \beta S_{\alpha,\beta}(\psi_{\chi}) = \frac{[f(\alpha)]^\beta-1}{1-\alpha}.
\end{equation}
To complete the proof of Eq.~(\ref{eq:Tsallis-order}), it suffices to show that \(g_\beta(\alpha)\) is monotonically nonincreasing. 
Differentiating yields
\begin{equation}
g_\beta'(\alpha) = \frac{(1-\alpha)h_\beta'(\alpha) + \bigl([f(\alpha)]^\beta-1\bigr)}{(1-\alpha)^2},
\end{equation}
where \(h_\beta'(\alpha)=\frac{d}{d\alpha}[f(\alpha)]^\beta=\beta[f(\alpha)]^{\beta-1}f'(\alpha)\).
Since \([f(\alpha)]^\beta\) is convex (this follows directly from the fact that its second derivative is nonnegative), the first-order condition gives~\cite{boyd2004convex}
\begin{equation}
[f(1)]^\beta \geq [f(\alpha)]^\beta + h_\beta'(\alpha)(1-\alpha).
\end{equation}
This implies
\begin{equation}
1\geq[f(\alpha)]^\beta+h_\beta'(\alpha)(1-\alpha)
\;\;\;\Longrightarrow\;\;\;
(1-\alpha)h_\beta'(\alpha) + \bigl([f(\alpha)]^\beta-1\bigr) \leq 0,
\end{equation}
and hence \(g_\beta'(\alpha) \leq 0\).
\sloppy
\bibliographystyle{unsrt}
\bibliography{main}
\fussy
\end{document}